\documentclass[11pt]{article}
\usepackage[margin=1in]{geometry}
\usepackage[english]{babel}
\usepackage{xr}
\usepackage[font={footnotesize}, labelfont={footnotesize, bf}]{caption}

\usepackage{setspace}

\usepackage{tikz}
\usetikzlibrary{arrows}

\usepackage{mathrsfs,hyperref, algorithm}
\usepackage[]{amsmath}
\usepackage{amsthm}
\usepackage{fix-cm}
\usepackage[]{amssymb}
\usepackage[]{latexsym}
\usepackage[right]{eurosym}
\usepackage[T1]{fontenc}
\usepackage[]{graphicx}
\usepackage[]{epsfig}
\usepackage{fancyhdr}
\usepackage{pstricks}
\usepackage{multirow}
\usepackage{subcaption}
\usepackage{cases}

\numberwithin{equation}{section}

\makeatletter

\DeclareMathSymbol{\shortminus}{\mathbin}{AMSa}{"39}

\newcommand{\bbr}{\mathbb{R}}
\newcommand{\E}{\mathbb{E}}
\newcommand{\bbe}{\mathbb{E}}
\newcommand{\bbn}{\mathbb{N}}

\newcommand{\bbp}{\mathbb{P}}

\newcommand{\bbd}{\mathbb{D}}

\newcommand{\fcal}{\mathcal{F}}

\newcommand{\ncal}{\mathcal{N}}

\newcommand{\acal}{\mathcal{A}}

\newcommand{\loss}{\mathbf{L}}

\newcounter{modcount}
\newcommand{\modulo}[2]{%
\setcounter{modcount}{#1}\relax
\ifnum\value{modcount}<#2\relax
\else\relax
\addtocounter{modcount}{-#2}\relax
\modulo{\value{modcount}}{#2}\relax
\fi}
\newcommand{\tablepictures}[4][c]{\begin{tabular}[#1]{@{}c@{}}#2\vspace{0.5cm}\\(\alph{#4}) #3\end{tabular}}
\newcounter{gridsearch}
\newcommand{\tabpic}[2]{
    \stepcounter{gridsearch}
    \modulo{\thegridsearch}{2}
    \ifnum\value{modcount}=0
        \tablepictures[t]{#1}{#2}{gridsearch}\\[2.0cm]
    \else
        \tablepictures[t]{#1}{#2}{gridsearch}&~&
    \fi
}

\makeatother
\newtheorem{lemma}{Lemma}[section]
\newtheorem{proposition}[lemma]{Proposition}
\newtheorem{theorem}[lemma]{Theorem}

\newtheorem{example1}[lemma]{Example}
\newtheorem{rem1}[lemma]{Remark}
\newtheorem{assumption}[lemma]{Assumption}
\newtheorem{alg1}[lemma]{Algorithm}
\newtheorem{me1}[lemma]{Mechanism}

\newenvironment{remark}{\begin{rem1}\rm}{\end{rem1}}

\usepackage{color}

\begin{document}
	
	\title{Dynamic Default Contagion in Heterogeneous Interbank Systems}
	\author{Zachary Feinstein\thanks{Stevens Institute of Technology, School of Business, Hoboken, NJ 07030, USA. \tt{zfeinste@stevens.edu}} \and Andreas S{\o}jmark\thanks{Imperial College London, Department of Mathematics, London, SW7 2AZ, UK. \tt{a.sojmark@imperial.ac.uk}}}
	\date{\today}
	\maketitle
	\abstract{
		In this work we provide a simple setting that connects the structural modelling approach of Gai--Kapadia interbank networks with the mean-field approach to default contagion.
		To accomplish this we make two key contributions.  First, we propose a dynamic default contagion model with endogenous early defaults for a finite set of banks, generalising the Gai--Kapadia framework.  Second, we reformulate this system as a stochastic particle system leading to a limiting mean-field problem.  We study the existence of these clearing systems and, for the mean-field problem, the continuity of the system response.
		~
	}

	\section{Introduction}\label{sec:intro}
	
	More than a decade after the collapse of Lehman Brothers and the threat of contagious defaults throughout the global financial system in 2008, systemic risk is still of vital importance to study.  Systemic risk is the risk of financial contagion, i.e., when the distress of one institution spreads due to interlinkages in balance sheets.  
	In this work we focus on the modelling of default contagion, which occurs if the failure of one bank or institution to repay its debts in full causes other defaults, triggering a chain reaction of failing banks.  This may occur through a network of interbank obligations as studied in the seminal works of~\cite{EN01, GK10, RV13} in a static, network-based setting.  More specifically, in those works, the default of a bank causes direct impacts to the balance sheets of other banks in the financial system. This loss of capital can cause other banks to default, thus spreading the original shock further throughout the system.
	
	Most work on network models for default contagion have focused on a static setting.  However, banking balance sheets are highly dynamic and subject to fluctuations due to, e.g., market movements.
	Indeed the conclusion of \cite{EN01} gives a discussion of how to include multiple clearing dates and time dynamics, which is studied in \cite{CC15,ferrara16}. Additionally, \cite{KV16} considers a similar approach to a financial model with multiple maturities.
	Most prior works on dynamic network models consider a discrete time setting~\cite{CC15,ferrara16,KV16}.  As far as we are aware, the only two extensions of the Eisenberg--Noe framework~\cite{EN01} to continuous time are~\cite{BBF18,sonin2017}. In this work we are interested in such a continuous time representation with early defaults which neither~\cite{BBF18,sonin2017} allow.

	Moreover, we are interested in making a precise connection between this dynamic balance sheet framework and the recent probability literature on mean-field approaches to contagion modelling \cite{HLS18, HS18, LS18a, NS17, NS18}. This, of course, also links to the financial mathematics literature on mean-field models for systemic risk, as exemplified by \cite{bo_capponi,capponi_clusters, fouque2013stability}, although these works have focused on `flocking effects' modelled by mean-reversion rather than default contagion. When formulating their models, both these strands of literature abstract away any considerations of the precise financial assumptions and underlying balance sheet mechanisms. This leaves an important gap that we aim to address with this work.
    Comparably to the mean-field limit considered herein, \cite{bo2014bilateral} presented the bilateral credit valuation adjustment for credit default swaps referencing a large number of entities; in particular, that work demonstrates an explicit characterization of the weak limit as the number of reference entities approaches infinity. In contrast, we are interested in sending the number of institutions in the financial network to infinity.

	Our main contribution is two-fold.  First, as detailed in Section~\ref{sec:dynamic}, we propose a dynamic default contagion model with early defaults driven by insolvency that generalises the Gai--Kapadia framework~\cite{GK10} under recovery of face value with historical price accounting (explained in detail below).  An analogous model for early defaults driven by illiquidity is provided in Supplemental~\ref{sec:illiquid}. We then prove existence of a greatest clearing solution in this financial setting and provide a discussion of how defaults spread through the system of banks (Prop.~\ref{prop:exist}). By itself, the introduction of \emph{early defaults} in this dynamic balance sheet based network model is novel in the literature, see also the related framework of \cite{Lipton2016}. Second, as detailed in Section~\ref{sec:mean-field}, starting from our \emph{finite bank} setting, we propose a tractable framework for reformulating the model as a \emph{stochastic particle system}, giving the greatest clearing capital solution when the contagion is governed by a suitable `cascade condition' (Prop.~\ref{prop:particle}). This then leads us to a limiting \emph{mean-field problem}, yielding a succinct  representation of the system, which is shown to evolve continuously in time under a constraint on the interactions (Thm.~\ref{thm_cont}). If this constraint is violated, jumps may survive the passage to the limit, and we end the paper by discussing a possible characterisation of such jumps in terms of a mean-field analogue of the finite cascade condition; see \eqref{eq:limit_PJC}.
	As such, this work provides an intriguing starting point for reconciling the structural (finite bank) balance sheet literature with the recent dynamic continuous-time mean-field approaches for default contagion.

	\section{Dynamic balance sheets and finite bank model}\label{sec:dynamic}
	
	In order to study defaults prior to maturity, we need to consider the valuation of any obligations due \emph{after} the default event.  Within the corporate debt literature, three primary notions of such valuation are considered (see, e.g.,~\cite{guo2008distressed}).
	Under \textbf{recovery of face value [RFV]}, in case of a default event, holders of bonds (of the same issuer) recover a fraction of the face value of the held bonds regardless of maturity. In the (static time) systemic risk literature, this formulation corresponds with the default contagion proposed in, e.g., \cite{GK10}.  The other notions are recovery of treasury [RT] (recovery of a fraction of the present value of a risk-free bond with the same maturity; this is equivalent to RFV under risk-free rate of $r = 0$) and recovery of market value [RMV] (recovery of a fraction of the value of the bond from just prior to default; this formulation corresponds with the default contagion process of~\cite{EN01,RV13} assuming a notion of network valuation adjustments as in, e.g., \cite{BF18comonotonic,barucca2016valuation}).
	In this work we focus on RFV as it is often considered more accurate than RMV in approximating the realised recovery rates for corporate bonds \cite{guha2020structural, guo2008distressed}.
	
	As hinted with the notions of valuation in default (especially RMV), we also need some notion of valuing bonds that have not yet matured.  This has been studied with a network valuation adjustment by, e.g., \cite{BF18comonotonic,barucca2016valuation}, in a single time step setting with perfect information of the entire financial system;~\cite{bo2014bilateral}, similarly, studies the bilateral credit valuation adjustment problem for credit default swaps.  Herein we assume only limited information is known by each bank about its counterparties.  Therefore in order to determine the probability of a \emph{future} default event only historical information about those obligations is known; as such \textbf{historical price accounting} will be utilised.  That is, the default probability of any institution is assumed based on its historical rate (e.g., based on the credit rating of a bond).
    That is, the banking book assumes that there is no probability of default regardless of the wealth of its counterparties.  This is due to the limited information available to each firm, i.e., the absence of an active market with which to mark interbank assets forces banks to consider the historical rate of default instead for accounting purposes.
    Within this work, to simplify matters, we will assume the historical default probability is 0 and, thus, this probability is utilised in the banking book until the default realises. 
    This is an optimistic accounting rule for marking interbank assets as it does not assume any updating of the historical default probability; this can be viewed as a best-case analysis of default contagion as it provides the greatest possible value for interbank assets and thus provides a bound on any other valuation system.
    The partial information available to each bank in the system implies, additionally, that banks can\emph{not} anticipate the default of their counterparties, but must wait until the event is declared in order to update their balance sheet.  This can precipitate a financial shock which can cascade through the financial system as will be detailed below.
	
	\begin{figure}[t]
		\centering
		\begin{subfigure}[t]{0.42\textwidth}
			\centering
			\begin{tikzpicture}[x=\linewidth/4.9,y=5mm]
			\draw[draw=none] (0,9.5) rectangle (5,10) node[pos=.5,yshift=0.3em]{};
			\draw[draw=none] (0,9) rectangle (3,9.5) node[pos=.5,yshift=0.2em]{\small \bf Assets};
			\draw[draw=none] (3,9) rectangle (5,9.5) node[pos=.5,yshift=0.2em]{\small \bf Liabilities};
			
			\filldraw[fill=blue!20!white,draw=black] (0,6.5) rectangle (3,9) node[pos=.5,style={align=center}]{\footnotesize External (Mark-to-Market) \\ ${\scriptstyle\bbe[x_i(T) \; | \; \fcal_t]}$};
			\filldraw[fill=yellow!20!white,draw=black] (0,4) rectangle (3,6.5) node[pos=.5,style={align=center}]{\footnotesize Interbank (Solvent) \\ ${\scriptstyle\sum_{j \in \acal_t} L_{ji}(T)}$};
			\filldraw[fill=orange!20!white,draw=black] (0,0) rectangle (3,4) node[pos=.5,style={align=center}]{\footnotesize Interbank (Insolvent) \\ ${\scriptstyle\sum_{j \in \ncal \backslash \acal_t} \Bigl(\begin{array}{l}{\scriptstyle (1 - R) L_{ji}(\tau_j)}\\ {\scriptstyle \;+ \,R L_{ji}(T)}\end{array}\Bigr)}$};
			
			\filldraw[fill=red!20!white,draw=black] (3,3) rectangle (5,9) node[pos=.5,style={align=center}]{\footnotesize Total \\ ${\scriptstyle\sum_{j \in \ncal_0} L_{ij}(T)}$};
			\filldraw[fill=green!20!white,draw=black] (3,0) rectangle (5,3) node[pos=.5,style={align=center}] (t) {\footnotesize Capital \\ ${\scriptstyle K_i(t)}$};
			\end{tikzpicture}
			\caption{\footnotesize Stylised balance sheet for bank $i $ at time $t$.}
			\label{fig:balance-sheet}
		\end{subfigure}
		~
		\begin{subfigure}[t]{0.47\textwidth}
			\centering
			\includegraphics[width=0.73\linewidth]{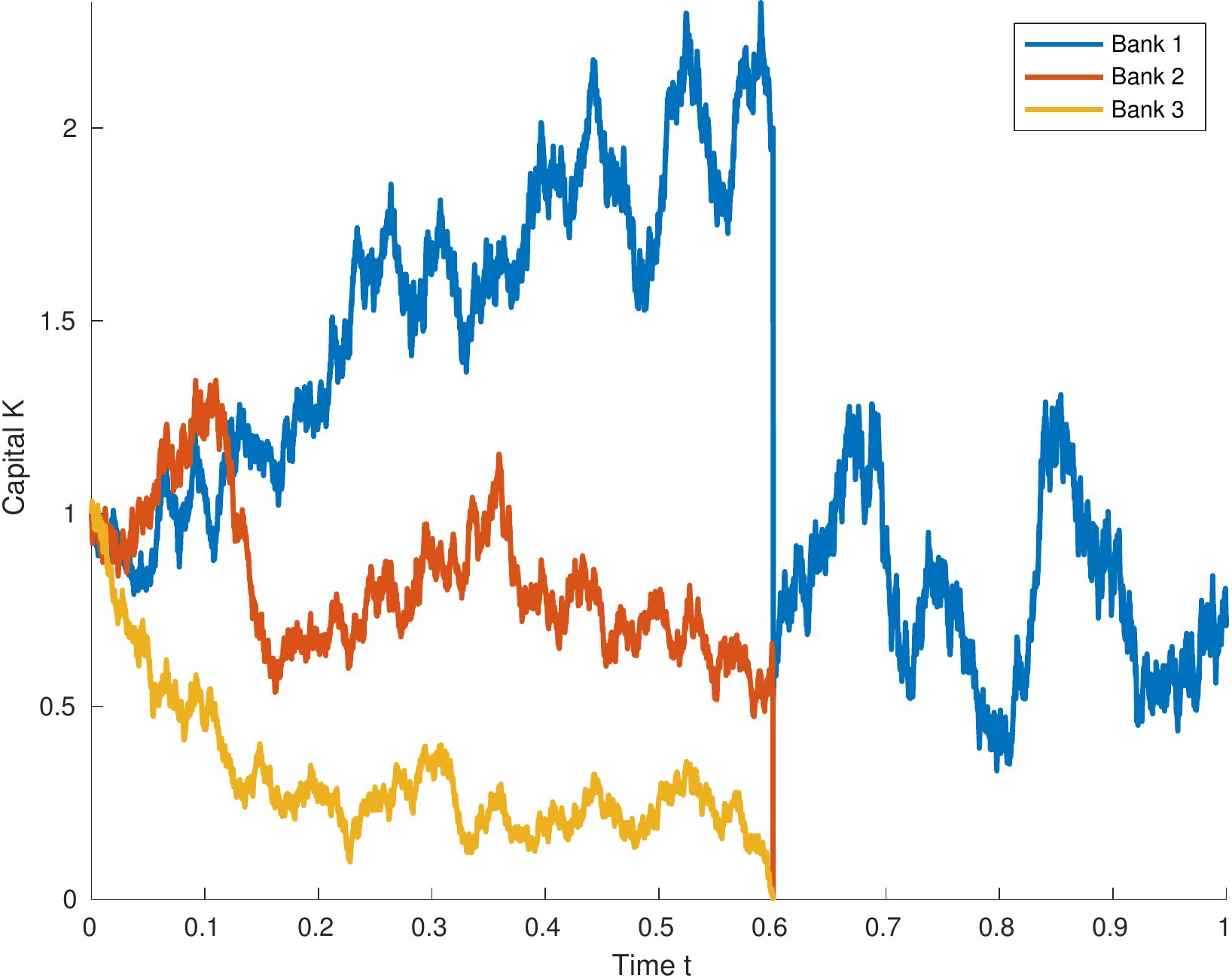}
			\caption{\footnotesize Realisation of a 3 bank system driven by correlated GBMs with a contagious default at $t \approx 0.6$.}
			\label{fig:3bank}\vspace{-4pt}
		\end{subfigure}
		\caption{Stylised balance sheet at time $t$ and a realisation of this system over time.}\vspace{-8pt}
	\end{figure}
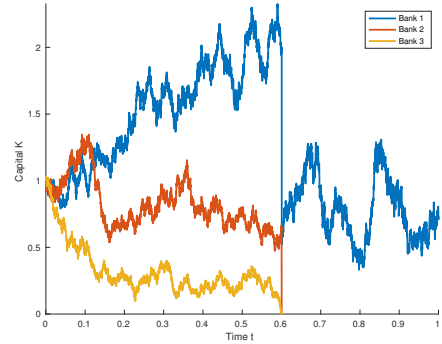
	
	\paragraph{2.1 Balance sheet construction}
	Throughout this section we will consider a financial system with $n \in \bbn$ financial institutions. This system does not include the central bank or any other non-financial entities; we will consider such entities, aggregated into a node $0$.  Notationally, $\ncal = \{1,2,...,n\}$ is the set of banks and $\ncal_0 = \ncal \cup \{0\}$ includes any external entities.
	
	In order to construct a continuous-time model with maturity $T < \infty$ we will begin by considering the stylised balance sheet for a generic bank $i$ in our system.  This balance sheet should be viewed as a dynamic and \emph{stochastic} version over times $[0,T]$ of the static setting of \cite{GK10} with filtered probability space $(\Omega,\fcal,(\fcal_t)_{t \in [0,T]},\bbp)$ with risk-neutral measure $\bbp$.  Throughout time, all assets are of only two types: interbank assets and external assets; all liabilities are either interbank (and thus assets for another bank $j$ in the system) or external and owed to node $0$.
	For simplicity, we will set the risk free rate to 0, i.e., $r = 0$.
	This stylised balance sheet is depicted in Figure~\ref{fig:balance-sheet}.
	
	Specifically, we will consider the book values for bank $i \in \ncal$.  Let $x_i(T) \in L^2(\fcal_T)$ be the value of the external assets for this bank at maturity.  The value of this external asset at time $t \in [0,T]$ is the discounted (risk-neutral) expectation of the value at maturity, i.e., $\bbe[x_i(T) \; | \; \fcal_t]$.  Throughout this work we will take the external asset processes $x_i$ to follow a non-negative (It\^o) process; if the drift component of $x_i$ is constructed from the risk-free rate $r = 0$ then $x_i(t) = \bbe[x_i(T) \; | \; \fcal_t]$, otherwise these values may differ.
	
	In contrast, the liabilities (and thus also interbank assets) have fixed and known (i.e., deterministic) coupon payments.\footnote{From the perspective of RFV, though we call these obligations coupon payments, each of these obligations will be treated as a separate zero-coupon bonds with different maturities.}  Notationally, let $L_{ij}(t)$ be the sum total of all coupon payments owed to bank $j \in \ncal_0$ up through time $t$ (i.e., on $[0,t]$).  In this way, $L_{ij}(T)$ denotes the total obligations through the maturity time $T$ owed to bank $j$.  For simplicity, we will additionally assume the interbank liabilities to be continuous in time, i.e., $t \in [0,T] \mapsto L_{ij}(t)$ is continuous.
	
	When valuing interbank assets, there are two settings to consider: \emph{old} obligations and \emph{future} coupon payments.  Consider a fixed time $t \in [0,T]$ as the ``present,'' i.e., all obligations $L_{ji}(t)$ from bank $j$ to bank $i$ are considered old whereas the obligations $L_{ji}(T)-L_{ji}(t)$ due after $t$ will be considered future obligations.
	We follow the simple assumption that every solvent firm makes all its current (and past) coupon payments on schedule.  To simplify notation, let $\acal_t$ denote the set of banks that are solvent at time $t$ (to be discussed further below).  If bank $j \in \acal_t$ is solvent at time $t$, then the value of the interbank assets $L_{ji}(t)$ for bank $i$ is exactly $L_{ji}(t)$ as these payments have been actualised.
    \begin{remark}\label{rem:insolvency}
    Herein we assume there is a functioning and well-capitalised \emph{external} repurchase agreement market so that no solvent firm is subject to a cash shortfall; this assumption is relaxed in, e.g.,~\cite{BBF18} or Supplemental~\ref{sec:illiquid}.\footnote{In practice, this repurchase agreement market would generally involve the financial firms themselves.  As such, raising liquidity would result in altered exposures $L$.  An interesting approach to this problem would be to allow for strategic behavior as in \cite{NS18}.}
    That is, any bank with a liquidity shortfall, but \emph{positive} capital, is able to borrow using its future assets as collateral.  Due to the collateralized and short-term nature of these repurchase agreements, we simplify to assume that such an agreement does not impose any additional costs on the bank.  It is, therefore, clear that a bank with negative cash account and negative net worth would be unable to cover its obligations and would default.  In fact, even if the cash account is positive, the shareholders of a bank would choose to default once the capital is no longer positive as these shareholders would find their financial prospects worse by paying off their obligations than by declaring bankruptcy and defaulting on all current obligations.  Thus, with the functioning and well-capitalised repurchase agreement market, banks will default exactly when their capital is no longer positive.
    We wish to draw the reader's attention to Supplemental~\ref{sec:illiquid} for details on defaults due to \emph{illiquidity} when, e.g., the repurchase agreement market is no longer functioning (as occurred in the 2008 financial crisis~\cite{GM09,B09}).
    \end{remark}
	Consider now bank $j \in \ncal \backslash \acal_t$ insolvent which defaults on its obligations at some time $\tau_j \leq t$.  As stated already all payments $L_{ji}(\tau_j)$ due before default are paid (and therefore marked) in full.  Following the RFV paradigm, the obligations $L_{ji}(t)-L_{ji}(\tau_j)$ are only partially repaid based on the recovery rate $R \in [0,1]$.  We assume this fixed recovery of prior defaulted obligations since defaults usually cause a delay in payments; this delay in the actualisation of payments (due to, e.g., bankruptcy courts) during the period of interest $[0,T]$ prompts us to assume the RFV paradigm rather than a clearing system such as studied in~\cite{EN01,RV13}.
	
	It remains to consider the value of \emph{future} interbank coupon payments.  Unlike the external assets $x_i$, these interbank assets $L_{ji}$ are \emph{nonmarketable}, i.e., there is no market in which to determine the value of the risk of default for these future coupon payments.  In fact, each bank only has direct information about its own assets and liabilities.  As such, due to this asymmetry of information and without a market to mark these interbank assets, we assume that the banks follow a historical price accounting rule for valuing these assets.  That is, all future coupon payments owed from solvent firms are marked in full until an actualised default event occurs.  Once a default event occurs, following the RFV paradigm presented above, all future obligations are immediately marked down by the recovery rate $R \in [0,1]$.
	
	Therefore, in total, the value of interbank assets $L_{ji}$ depends on the default time of bank $j$.  If we set the default time of bank $j$ to $\tau_j$ (to be discussed in greater detail below), interbank assets $L_{ji}$ are marked at time $t$ as: $L_{ji}(T)$ if $t < \tau_j$, and $L_{ji}(\tau_j) + R[L_{ji}(T)-L_{ji}(\tau_j)]$ if $t \geq \tau_j$ (equivalently $j \in \acal_t$ and $j \in \ncal \backslash \acal_t$ respectively).
	
	The stylised balance sheet encoding these assets and liabilities at time $t \in [0,T]$ is displayed in Figure~\ref{fig:balance-sheet}.  The realised capital $K_i(t)$ at time $t \in [0,T]$ for bank $i$ is, thus, the difference between the value of the assets and liabilities.  That is, given the set of solvent banks $\acal_t \subseteq \ncal$ at time $t$, the realised capital for bank $i$ is computed as:
	\begin{equation}\label{eq:Kapital}
	K_i(t) := \E[x_i(T) \; | \; \fcal_t] + \sum_{j \in \acal_t} L_{ji}(T) + \sum_{j \in \ncal \backslash \acal_t} \left((1-R)L_{ji}(\tau_j) + R L_{ji}(T)\right) - \sum_{j \in \ncal_0} L_{ij}(T).
	\end{equation}
	As in the static framework of~\cite{GK10}, bank $i$ is deemed to be in default once its assets are worth less than its total liabilities, i.e., when it has negative realised capital.
	That is, insolvency for bank $i$ occurs at the stopping time $\tau_i = \inf\{t \in [0,T] \; | \; K_i(t) \leq 0\}$.
	This is consistent with the notion that there exists a repurchase agreement market external to the system of banks as briefly mentioned previously; if a bank's realised capital is positive then it can raise the necessary cash in order to remain solvent, but it cannot raise sufficient cash to pay its obligations if it has negative realised capital.  
	The set of solvent firms at time $t$ is, formally, given by $\acal_t := \{i \in \ncal \; | \; \tau_i > t\}$.
	
	\begin{assumption}\label{ass:balance-sheet}
		The modelling assumptions expressed above can be summarised thusly:
		\begin{enumerate}
			\item the external assets of each bank follow a stochastic process (which can be correlated to each other) and (being marketable) are marked-to-market with risk-neutral measure $\bbp$;
			\item the interbank assets and liabilities are solely based on contracts written prior to time $0$ and have fixed coupon schedule;
			\item interbank assets (being nonmarketable) are valued using historical price accounting, i.e., priced at face value prior to a default event and reevaluated with RFV after default; and
			\item defaults occur once the realised capital of a bank drops below zero.
		\end{enumerate}
        Defaults come as a shock to the system and cause a jump in the capital of any connected institution due to the partial information available for the accounting of interbank assets; in a single-firm context, this is made explicit in~\cite{duffie2001term}.
	\end{assumption}
	The shocks due to default outlined above are realistic since the interbank assets are nonmarketable.  If, however, banks attempted a counterparty or network valuation adjustment (see, e.g., \cite{BF18comonotonic, barucca2016valuation}) default shocks would still be expected due to the asymmetric and incomplete information available to the different banks.  
	
	\paragraph{2.2 Clearing capital}
	Under the setting summarised in Assumption~\ref{ass:balance-sheet}, the realised capital for bank $i$ can be considered as an equilibrium setting that depends directly on the defaulting events of all other banks.  Explicitly, this is provided by~\eqref{eq:Kapital} where the set of solvent firms $\acal_t(K)$ and default times $\tau(K)$ are considered directly as functions of the vector of wealth processes $K: \Omega \times [0,T] \to \bbr^n$.
	The existence of a clearing solution to this equilibrium problem is provided in the following proposition.
	\begin{proposition}[Clearing capital]\label{prop:exist}
        (i) There exists a greatest and least clearing capital $K^\uparrow \geq K^\downarrow$ (component-wise and a.s.~for every time $t$) to the network clearing problem defined by~\eqref{eq:Kapital}. (ii) Any clearing capital solution $K$ to the network clearing problem defined by~\eqref{eq:Kapital} is c\`adl\`ag.
	\end{proposition}
	\begin{proof}
       (i) First, note that the fixed point problem for the capital process $K$ always maps into the complete lattice (with component-wise and almost sure ordering for every time $t$) $\prod_{t \in [0,T]} \bbd_t$ for
		\[\bbd_t := \bbe[x(T) \; | \; \fcal_t] + \sum_{j \in \ncal} [R,1] \times L_{j\cdot}(T) - \sum_{j \in \ncal_0} L_{\cdot j}(T).\]
		Second, note that capital process $K$ only depends on itself through the default times $\tau(K)$ (and the set of solvent firms $\acal_t(K) = \{i \in \ncal \; | \; \tau_i(K_i) > t\}$).  Clearly, by definition, as the capital process $K$ decreases the default times $\tau(K)$ do \emph{not} increase.  Further, as banks default, the entire system's wealth drops as well since $R \leq 1$.  With this, we are able to complete this proof through an application of Tarski's fixed point theorem.

     (ii) Let $K$ denote an arbitrary clearing capital with associated default times $\tau_i = \inf\{t \in [0,T] \; | \; K_i(t) \leq 0\}$.  Trivially, by~\eqref{eq:Kapital}, $t \in (\tau_{[k]},\tau_{[k+1]}) \mapsto K_i(t)$ is continuous for any $k \in \{0,1,...,n\}$ where $\tau_{[k]}$ denotes the $k^{th}$ order statistic ($0 =: \tau_{[0]} \leq \tau_{[1]} \leq \tau_{[2]} \leq ... \leq \tau_{[n]} \leq \tau_{[n+1]} := T$ almost surely).  It remains to prove that $K_i(t)$ is right-continuous at all default times.  Consider $k \in \{0,1,...,n\}$ such that $\tau_{[k]} < \tau_{[k+1]}$.  By construction, the set of solvent banks $\acal_t = \acal_{\tau_{[k]}}$ a.s.\ for every $t \in [\tau_{[k]},\tau_{[k+1]})$.  Therefore, for $t \in [\tau_{[k]},\tau_{[k+1]})$, the capital of bank $i$ satisfies
        \[K_i(t) = \E[x_i(T) \; | \; \fcal_t] + \sum_{j \in \acal_{\tau_{[k]}}} L_{ji}(T) + \sum_{j \in \ncal \backslash \acal_{\tau_{[k]}}} ((1-R)L_{ji}(\tau_j) + R L_{ji}(T)) - \sum_{j \in \ncal_0} L_{ij}(T)\]
        with the \emph{fixed} set of solvent firms.  This is clearly right continuous as $K_i(\tau_{[k]}) = \lim_{t \searrow \tau_{[k]}} K_i(t)$ for every bank $i$.
	\end{proof}
	As is typical in the literature, we will primarily focus on the greatest clearing solution $K^\uparrow$.  Briefly, we will discuss how to use a \emph{fictitious default algorithm} to compute this clearing solution forward in time.  Such an algorithm assumes that at time $t \in [0,T]$, any bank that was solvent prior to $t$ ($\acal_{t\shortminus}$) is assumed a priori to still be solvent; this is the best case scenario for all banks due to the downward stresses from a default.  Solvency ($K_i(t) > 0$) of all banks is then checked under this scenario; if no banks default we can move forward in time, otherwise any new defaults may cause a domino effect of further defaults.  In the case of defaults, we update the balance sheet of all solvent firms to determine if this shock causes a cascade of failures. This sequential testing for new defaults and updating the balance sheets continue until no new defaults occur.  In practice this algorithm is run using an event finding algorithm to determine the time of the initial default, at that time the cascading defaults are determined until the system re-stabilises at a new set of solvent institutions, and the stochastic processes evolve normally until the next default event.  This is demonstrated in Figure~\ref{fig:3bank} where the insolvency of one bank causes another bank to default as well.  If desired, the least clearing solution $K^\downarrow$ could be found analogously with a fictitious solvency algorithm instead.

	\section{Mean-field version of the model}\label{sec:mean-field}

Recall that $L_{ij}(T)$ denotes the liabilities for the full period $[0,T]$. 
To have a simple model for repayment, we will assume that there is a non-decreasing function $\psi(T,\cdot):[0,T]\rightarrow[0,\infty)$ and constant relative liabilities $\lambda_{ij}\geq0$ such that the liabilities of bank $i$ owed for the remaining period $[t,T]$ are given by
\begin{equation}\label{liab_dyn}
L_{ij}(T)-L_{ij}(t)=\psi(T,t) \lambda_{ij}
\end{equation}
with $L_{ij}(0)=0$, for $i,j\in \mathcal{N}_0$, where $\lambda_{ij}=0$ if $i=j$ or $i=0$. For example, we could consider a linear repayment schedule so that $\psi(T,t)=T-t$.

Let us write $\lambda_i^\mathrm{ext}:=\lambda_{i0}$ for the external liabilities. By using \eqref{liab_dyn} and reorganising the sums in equation \eqref{eq:Kapital} for the capital $K_i(t)$, we find that it satisfies
\begin{equation}\label{eq:K_i}
K_i(t)= E[x_i(T) \; | \; \fcal_t] - \psi(T,0)\bigl( \lambda_{i}^\mathrm{ext} + \sum_{j \in \ncal} [\lambda_{ij} - \lambda_{ji}]\bigr)- (1-R)\!\!\!\sum_{j \in \ncal\backslash\acal_t}\! \psi(T,\tau_j)\lambda_{ji},\vspace{-2pt}
\end{equation}
where we recall that $\mathcal{N}$ is the full set of banks while $\acal_t$ is the set of solvent banks at time $t$. This coupled system is the focus of the rest of the paper, where we will further assume that the stochastic processes $x_i(t)$ are given by geometric Brownian motions
\begin{equation}\label{GBM}
dx_i(t) = x_i(t)[\mu_i(t) dt + \sigma_i(t)] dW_i(t) \quad \text{with}\quad W_i(t)=\sqrt{1-\rho^2}B_i(t)+\rho B_0(t),
\end{equation}
for independent Brownian motions $B_0,\ldots,B_n$, and continuous functions $\mu_i,\sigma_i$.

\paragraph{3.1 Tractable asymmetry}

As has been observed in several empirical studies, interbank networks often display pronounced core-periphery features with negligible periphery-to-periphery interactions \cite{CP14,FL15, veld2014core}. This type of network structure for the relative liabilities $\lambda_{ij}$ can be conveniently captured by the framework we introduce next. Naturally, our framework is an idealisation of reality,
but it leads to a tractable model, and we stress that interbank liabilities are not fully observable in practice so approximations will always be involved.

In short, one may reasonably hope to explain the heterogeneity of the relative liabilities by only a low number of underlying characteristics, say $k \ll n$, for a system with $n$ banks. Perhaps a little surprisingly at first, we can have a rich and practically relevant model already with $k=1$. Indeed, similarly to \cite{FL15}, we can take as our one characteristic a global score for the `coreness' of each bank, so that the relative liabilities are given by $\lambda_{ij}=u^iv^j$ for $i\neq j$ (and $\lambda_{ij}=0$ for $i=j$), where $u^i\geq 0$ is a score for `how core' bank $i$ is in terms of borrowing and $v^j \geq 0$ is a score for `how core' bank $j$ is in terms of lending. In practice, these scores $u^i$ and $v^j$ could, e.g.,~be obtained by generating samples of the partially observed liabilities and then minimising the squared off-diagonal deviations $\sum_{i\neq j}(\lambda_{ij}-u^i v^j)^2$.

Depending on the task at hand, we may want a higher level of granularity than a global score for `coreness'.  As discussed in~\cite{CP14} financial networks typically have a `tiering' order, whereby \emph{top-tier} banks lend to each other and \emph{lower-tier} banks, while \emph{lower-tier} banks do \emph{not} lend to each other but instead lend to \emph{top-tier} banks. This suggests an organisation of the liability network according to (i) what tier and subgroup thereof the `borrower' belongs to, and (ii) how the lending of the `lender' is spread out according to this.  Generalising the earlier $1$-dim coreness score, we propose to let (i) be summarised by a $k$-dim score $u^i=(u^i_1,\ldots, u^i_k)$ along $k$ given characteristics, and (ii) by a corresponding $k$-dim score $v^i=(v^i_1,\ldots, v^i_k)$ such that
\vspace{-3pt}\begin{equation}\label{lambda_fact}
\lambda_{ij}=u^i \cdot v^j =\sum_{l=1}^k u^i_l v^j_l \quad \text{for} \quad i\neq j,\qquad \text{and} \qquad \lambda_{ij}=0 \quad \text{for} \quad i=j.\vspace{-3pt}
\end{equation}
We refer the interested reader to Supplemental~\ref{sec:4example} for a four subgroup example illustrating this construction.

In place of \eqref{lambda_fact} serving as a model for the network structure, we note that one could also treat it as a form of principle component analysis of the network, similarly to the use of spectral decompositions for core-periphery detection \cite{porter_2016} and related financial contagion models \cite{amini_minca, spiliopoulos_2019}.

\paragraph{3.2 Stochastic particle system}

In line with real-world systems, we assume net-lenders in the interbank  market still have strictly positive net liabilities overall, i.e.,~$\lambda_{i}^\mathrm{ext} + \sum_{j \in \ncal} [\lambda_{ij} - \lambda_{ji}]>0$. Thus, $K_i(t)<E[x_i(T) \,| \, \fcal_t]$ in \eqref{eq:K_i}, which simply says that the capital is strictly less than the value of external assets, and so we can introduce the logarithmic \emph{distances-to-default}
\begin{equation}\label{dist-to-def}
X_i(t):=\log \Bigl(  \frac{E[x_i(T) \; | \; \fcal_t]}{  E[x_i(T) \; | \; \fcal_t]-K_i(t)}  \Bigr), \quad \text{for}\quad i\in \mathcal{N}.
\end{equation}
This allows us to describe the health of the financial system in the following way.
\begin{proposition}[Particle system]\label{prop:particle}
	Assuming \eqref{liab_dyn}, \eqref{GBM} and \eqref{lambda_fact}, the greatest clearing capital solution to \eqref{eq:K_i} corresponds to $X$ in \eqref{dist-to-def} being the unique c\`adl\`ag solution to
	\begin{equation}\label{particle_sys}
	\begin{cases}
	dX_i(t) =  -\tfrac{1}{2}\sigma_i^2(s)ds  +  \sigma_i(s)dW_i(t) -dF_{i}(t), \;\;\; \tau_i=\inf \{ t\geq 0 \; | \;  X_i(t) \leq 0  \},\\
	F_{i}(t)=\displaystyle \log\Bigl(1 + \frac{1\!-\!R}{\psi(T,0)\Lambda_i} \sum_{l=1}^{k} v_{l}^i \! \int_0^t \psi(T,s) d\mathcal{L}^n_{l,i}(s) \Bigr),\;\;\; \mathcal{L}^n_{l,i}(t) = \sum_{j \neq i }^{n} u_{l}^j \mathbf{1}_{t\geq\tau_j}  \\
X_i(0) = \displaystyle\log\Bigl( \frac{x_{i}(0)}{\psi(T,0)\Lambda_i} \Bigr) +\!\displaystyle\int_0^T\!\!\!\mu_{i}(t)dt,\;\;\; \Lambda_{i} = \lambda_{i}^\mathrm{ext} + \sum_{l=1}^k\sum_{j \neq i}^n( v^i_l u^j_l  - u^i_l  v^j_l ),
\end{cases}
\end{equation}
with default set $ \mathcal{D}_t :=\{ i \; | \; \tau_i = t  \}= \{i \; | \; X_i(t\shortminus)-\Theta^n(t;\Delta \mathbf{L}^n(t);i) \leq 0,\;\tau_i\geq t  \}$ and jump sizes $\Delta \mathcal{L}^{n}_{l,i}(t) = \Xi^n_l(t,\Delta \mathbf{L}^n) - u_l^i\mathbf{1}_{i\in \mathcal{D}_t}$ given by the following `cascade condition'
\begin{equation}\label{particle_cascade_cond1}
\begin{cases}
\Delta \mathbf{L}^n_{v}(t) = \displaystyle \lim_{m\rightarrow n} \Delta^{n,(m)}_{t,v}, \quad \Delta^{n,(0)}_{t,v}:= \Xi(t;0,v), \;\;\; \Delta^{n,(m)}_{t,v}:= \Xi^n(t, \Delta^{n,(m-1)}_{t,\cdot}, v ), \\
\Xi^n(t;f,v):= \displaystyle \sum_{l=1}^k v_l 	\Xi^n_l(t;f), \quad 	\Xi^n_l(t;f):=  \sum_{j= 1}^{n} u_{l}^j  \mathbf{1}_{ \{ X_j(t\shortminus) \in [0, \Theta^n(t; f;j) ]  , \, t\leq \tau_j \} } ,\\
\Theta^n(t;f;j) := \displaystyle \log\Bigl(1 + \frac{1\!-\!R}{\Lambda_{j}}\! \int_0^{t\shortminus} \!\frac{\psi(T,s)}{\psi(T,0)}d\mathbf{L}^n_{v^j}(s) + \frac{1\!-\!R}{\Lambda_{j}}\frac{\psi(T,s)}{\psi(T,0)}f(v^j) \Bigr)  - F_{j}(t\shortminus),\vspace{-2pt}
\end{cases}
\end{equation}
where $\mathbf{L}_v^n(t):=\sum_{j=1}^nv\cdot u^j_l\mathbf{1}_{t\geq \tau_j}= \sum_{l=1}^kv_l(\mathcal{L}^n_{l,i}(t)+u_l^i\mathbf{1}_{t\geq \tau^i})$.
\end{proposition}

The proof is found in Supplemental~\ref{proof:prop:particle}, but the intuition for \eqref{particle_sys}-\eqref{particle_cascade_cond1} is as follows. When bank $j$ defaults, each bank $i\neq j$ has a remaining exposure $\psi(T,\tau_j)v^i\cdot u^j$ of which only a fraction $R<1$ is recovered. Thus, the capital takes a hit of $(1-R)\psi(T,\tau_j)v^i\cdot u^j$. Since we are looking at distances-to-default, the dynamics in \eqref{particle_sys} are concerned with the evolution of capital over total liabilities, so the loss of capital translates to a downward jump inversely proportional to the net liabilities $\psi(0,T)\Lambda_i$. At first sight, this may seem to suggest that large external liabilities $\lambda^{\text{ext}}_i$ will dampen contagion, however, notice the dual effect of this:~all else equal, large values of $\lambda^{\text{ext}}_i$ also push the initial profile of the system much closer to the default barrier, hence making contagion more likely and its consequences more severe even for smaller downward jumps.

Looking at the contagion processes $\mathcal{L}^n_{l,i}$, the default of bank $j$ contributes a jump of $u_l^j$ to each $\mathcal{L}^n_{l,i}$. Thus, each bank $i \neq j$ feels an aggregate stress through the resulting jump in $\sum_{l=1}^kv_l^i\mathcal{L}_{l,i}^n$, where the weights $v_1^i,\ldots,v_k^i$ give the exposure of bank $i$ to the $k$ characteristics determining the network of liabilities. The cascade condition \eqref{particle_cascade_cond1} arrives at the size of these jumps iteratively, by starting from $\Delta^{n,(0)}_{t,v^i}$, which gives the contribution to $\Delta \mathcal{L}_{l,i}^n(t)$ solely from banks $j\neq i$ defaulting at time $t$ without considering contagion. Next, $\Delta^{n,(1)}_{t,v^i}$ takes into account any defaults resulting from a jump of size $\Delta^{n,(0)}_{t,v^i}$, and so on, comparably to the so-called fictitious default algorithm briefly discussed after Proposition \ref{prop:exist} in Section \ref{sec:dynamic}.

\paragraph{3.3 The mean-field model}
Given a financial system with $n$ banks, we artificially construct larger and larger systems of size $N=pn$ for $p\geq1$ with the same network structure. Specifically, for our given $n\geq 1$, we assume the pairs $\{(u^i,v^i)\}_{i=1}^n$ are drawn from some distribution $\hat{\varpi}$ on $\mathbb{R}^k\times\mathbb{R}^k$, and then we define\vspace{-3pt}
\begin{equation}\label{eq:lambda-N}
\lambda_{ij}^N:= \frac{n}{N} \hat{u}^i \cdot  \hat{v}^j = \frac{n}{N} \sum_{l=1}^{k}\hat{u}^i_l  \hat{v}^j_l, \quad\text{for}\quad  N \geq 1,\vspace{-2pt}
\end{equation}
by taking $\{(\hat{u}^i , \hat{v}^i) \}_{i=1}^\infty$ to be i.i.d.~samples from $\hat{\varpi}$. For any $N\geq1$, we are thus creating $p=N/n$ randomly sampled copies of each bank from the original $n$-bank system, each with their liability positions scaled down by $1/p=n/N$. Since the true network is not fully observable, we can think of this as choosing our approximate network structure \eqref{lambda_fact} by fixing a distribution $\hat{\varpi}$, e.g.,~based on generating samples of the true network. If one prefers a simple `multitype' structure with homogeneity within groups, one could alternatively start from \eqref{lambda_fact} and define $\hat{\varpi}$ to be the corresponding empirical measure. 

Let $\hat{\varpi}$ be a probability measure on $\mathbb{R}^k\times\mathbb{R}^k$ with the above interpretation, and restrict to heterogeneous parameters of the form $f_i=f_{u^i,v^i}$ for $f_i=\lambda^{\mathrm{ext}}_i,\mu_i,\sigma_i$. Sending $N\rightarrow \infty$ in the $N$-bank particle system \eqref{particle_sys} given by the interbank liabilities \eqref{eq:lambda-N}, we are then led to formulate the mean-field problem
\begin{equation}\label{CMV}
\begin{cases}
dX_{u,v}(t) =  -\tfrac{1}{2}\sigma^2_{u,v}(t)dt + \!\sigma_{u,v}(t) \bigl(\sqrt{1-\rho^2}dB(t)+\rho dB^0(t) \bigr) - d F_{u,v}(t), \vspace{3pt} \\
F_{u,v}(t)=\displaystyle\log \Bigl(  1 + \frac{1\!-\!R}{\Lambda_{u,v} } \displaystyle \sum_{l=1}^{k} v_l \!\int_0^t \frac{\psi(T,s)}{\psi(T,0)} d\mathcal{L}_l(s)   \Big),  \\
\mathcal{L}_l(t)  = \!\displaystyle\int_{\mathbb{R}^k\times\mathbb{R}^k}  \!u_l \mathbb{P}( t\geq \tau_{u,v} \, | \, B_0) d\hat{\varpi}(u,v), \;\;\; \tau_{u,v} =\inf \{ t\geq 0 \; | \; X_{u,v}(t) \leq 0 \},\\
X_{u,v}(0) =\displaystyle\log\Bigl( \frac{x_{u,v}(0)}{ \psi(T,0)\Lambda_{u,v}} \Bigr) +\int_0^T\!\!\!\!\mu_{u,v}(t)dt,\;\;\Lambda_{u,v} = \!\lambda^{\mathrm{ext}}_{u,v} + \displaystyle\sum_{l=1}^k ( v_l  \mathbb{E}[u_l]  \!- \!u_l  \mathbb{E}[v_l]  ),
\end{cases}
\end{equation}
where $B,B_0$ are independent Brownian motions. In this mean-field formulation, the heterogeneity is now modelled by the distribution $\hat{\varpi}$ of the pairs $(u,v)$. Formally, an `infinitesimal' bank indexed by $(u,v)$ has liabilities proportional to $u \cdot \hat{v}  =\sum_{l=1}^k  u_l \hat{v}_l$ towards infinitesimal banks indexed by $(\hat{u},\hat{v})$, where each $u_l\hat{v}_l $ gives the exposure to the $k$ characteristics determining the network structure. The value $\mathbb{P}( t\geq \tau_{u,v} \, | \, B_0)$ can be interpreted as the proportion of infinitesimal banks indexed by $(u,v)$ that have defaulted by time $t$. Aggregating the contagion, an infinitesimal bank indexed by $(\hat{u},\hat{v})$ has thus felt an accumulated stress proportional to $\hat{v}_l\mathcal{L}_l$ from its exposure to the $l$'th characteristic and a stress proportional to $\sum_{l=1}^k\hat{v}_l\mathcal{L}_l$ overall.

Under suitable assumptions, we demonstrate in a separate work \cite{fein-soj} that the McKean--Vlasov problem \eqref{CMV} indeed emerges as the mean-field formulation of the particle system \eqref{particle_sys}, when sending the number of banks to infinity. Here, we instead concentrate our efforts on introducing the mean-field model and deriving an intuitive condition that rules out singularities in the dynamics (see Theorem \ref{thm_cont} below).

\begin{proposition}[Convergence]\label{prop:convergence}
	Define the empirical measures $\mathbf{P}^N:=N^{-1}\sum_{i=1}^N\delta_{\hat{u}^i} \otimes \delta_{\hat{v}^i} \otimes \delta_{X^i}$ and $\mathbf{P}^N_0:=N^{-1}\sum_{i=1}^N\delta_{\hat{u}^i} \otimes \delta_{\hat{v}^i} \otimes \delta_{X^i(0)}$, and let Assumption \ref{assump:mean-field} hold. Under weak convergence induced by Skorokhod's M1 topology, every subsequence of $(\mathbf{P}^N)_{N\geq1}$ has a further subsequence converging in law to a random probability measure $\mathbf{P}^*$ such that
	\begin{equation}\label{eq:relaxedCMV}
	\langle \mathbf{P}^*, f\varphi \rangle = \int_{\mathbb{R}^k\times \mathbb{R}^k }f(u,v) \mathbb{E}[\varphi(X^*_{u,v})\,|\,B_0,\mathbf{P}^*]d\hat{\varpi}(u,v),
	\end{equation}
	for measurable $f:\mathbb{R}^k\times\mathbb{R}^k\rightarrow \mathbb{R}$ and  $\varphi:D_{\mathbb{R}}[0,T]\rightarrow \mathbb{R}$,
	where $X^*_{u,v}$ satisfies the dynamics in \eqref{CMV} with $\mathbb{P}( t \geq \!\tau_{u,v} \, |\, B_0)$ replaced by $\mathbb{P}( t\geq \!\tau^*_{u,v}  \,| \,B_0, \mathbf{P}^*)$ for $\tau^*_{u,v}=\inf\{t>0 : X_{u,v}^*\leq 0\}$. Under the conditions of Theorem \ref{thm_cont} below, there is a unique limit $\mathbf{P}^*$ and the additional conditioning on $\mathbf{P}^*$ can be dropped both in \eqref{eq:relaxedCMV} and in the dynamics of $X^*_{u,v}$.
\end{proposition}
The proof of Proposition \ref{prop:convergence} is given in \cite[Theorems 2.4 and 2.6]{fein-soj}.
Using the terminology of \cite{fein-soj,LS18a}, the limit points in Proposition \ref{prop:convergence} are relaxed solutions of \eqref{CMV}. 
\begin{assumption}[Structural conditions]\label{assump:mean-field}
	Let $f_{u,v}(t)$ be continuous in $(u,v,t)$, for $f=\sigma,\mu, \lambda^{\emph{\text{ext}}}$, and let $ \sigma_{u,v}(\cdot)$ be $\alpha$-H{\"o}lder continuous with $\alpha >1/2$. As in \eqref{eq:lambda-N}, let $\{(\hat{u}^i,\hat{v}^i)\}_{i=1}^\infty$ be i.i.d.~samples from $\hat{\varpi}$ with compact support such that $N^{-1}\sum_{i=1}^N\delta_{\hat{u}^i} \otimes \delta_{\hat{v}^i} $ converges weakly to $\hat{\varpi}$. Finally, let $x_{i}(0)=h(u^i,v^i,\xi^i)$ be a measurable function of $(u^i,v^i,\xi^i)$, for i.i.d.~samples $\{\xi^i \}_{i=1}^\infty$ such that $\mathbf{P}^N_0$ converges weakly to a probability measure $V_0(x|u,v)dxd\varpi(u,v)$ on $\mathbb{R} \times \mathbb{R}^k\times\mathbb{R}^k$, where $V_0(\cdot|u,v)$ is a continuous density vanishing at $0$.
\end{assumption}

While the finite system \eqref{particle_sys} has contagion occurring as jumps dictated by the cascade condition \eqref{particle_cascade_cond1}, one may hope for a smoother mean-field problem. Indeed, it turns out  we have a simple criterion for the mean-field to evolve continuously, meaning that contagion events are smoothed out in time. Intuitively, the default of an infinitesimal bank indexed by $(u,v)$ causes a stress proportional to $u \cdot \hat{v} \geq 0$ for banks indexed by $(\hat{u},\hat{v})$, so if the density of infinitesimal banks indexed by $(u,v)$ is sufficiently inversely proportional to this, the overall effect should be controlled. Before stating the result, we write $\hat{\varpi}=\mathrm{Law}(\mathbf{u},\mathbf{v})$ and let $S(\mathbf{u},\mathbf{v})$ denote its support, while we let $S(\mathbf{u})$ and $S(\mathbf{v})$ denote the support of $\mathbf{u}$ and $\mathbf{v}$, respectively. From here on we assume $S(\mathbf{u})$ and $S(\mathbf{v})$ are compact with $u\cdot v \geq 0$ for all $v\in S(\mathbf{v})$ and $u\in S(\mathbf{u})$ (as can be seen to follow from Assumption \ref{assump:mean-field}, noting that $\hat{u}^i \cdot \hat{v}^j = N \lambda_{ij}^N/n \geq 0$ in the finite system).
{A concrete example of such a system and related conditions based on a $k=4$ type network structure is provided in Supplemental~\ref{sec:4example}.

\begin{theorem}[Continuous mean-field]\label{thm_cont}
Let $V_0(\cdot|u,v)$ be the density of $X_{u,v}(0)$ in \eqref{CMV}. If
\begin{equation}\label{smallness_cond}
\Vert V_{0}(\cdot|u,v) \Vert_{\infty} < \frac{\Lambda_{u,v}}{1\!-\!R} \frac{1}{\max\{u\cdot \hat{v} \; | \;  \hat{v}\in S(\mathbf{v}) \;\mathrm{s.t.}\; u\cdot \hat{v} >0\}} \quad  \forall (u,v) \in S(\mathbf{u},
\mathbf{v}),
\end{equation}
then any a priori c\`adl\`ag solution to \eqref{CMV} is continuous in time (here $1/\max \emptyset=+\infty$).
\end{theorem}
\begin{proof}
Fix $t\geq0$. With $\Xi$ and $\Theta$ defined in \eqref{the_map_Xi} of Proposition \ref{fp_constraint} below, we can estimate
\vspace{-3pt}\begin{align*}\label{eq:the_map_Xi}
\Xi(t;f,v) & \leq \sum_{l=1}^k v_l \int_{\mathbb{R}^k\times\mathbb{R}^k} \hat{u}_{l}  \mathbb{P} \bigl( X_{\hat{u},\hat{v}}(t\shortminus ) \in [0, \Theta(t; f, \hat{v} ) ]  , \, t\leq \tau  \mid B_0\bigr) d\hat{\varpi}(\hat{u},\hat{v}) \\
&\leq  \max_{\tilde{v}\in S(\mathbf{v})} f(\tilde{v}) \int_{\mathbb{R}^k\times\mathbb{R}^k} v \cdot \hat{u}  \frac{1-R}{\Lambda_{\hat{u},\hat{v}}}\Vert V_0(\cdot | \hat{u},\hat{v})   \Vert_{\infty} d\hat{\varpi}(\hat{u},\hat{v}) < \max_{\tilde{v}\in S(\mathbf{v})} f(\tilde{v})
\end{align*}
for all $v\in S(\mathbf{v})$, by \eqref{smallness_cond}.
Now let $\mathbf{L}_v(t):= \sum_{l=1}^kv_l \mathcal{L}_l(t)$ and take $f(v):=\Delta \mathbf{L}_v(t)$. Then it follows from the above bound and \eqref{eq:fp_constraint} in Proposition \ref{fp_constraint} below that
\[
\max_{v\in S(\mathbf{v})} \Delta \mathbf{L}_v (t)  = \max_{v\in S(\mathbf{v})} 	\Xi(t;\Delta \mathbf{L},v)   < 	\max_{v\in S(\mathbf{v})} \Delta \mathbf{L}_v (t),
\]
where we have used the compactness of $S(\mathbf{v})$ along with continuity in $v$. Therefore, we deduce $ \Delta \mathbf{L}_\cdot(t)  \equiv 0$, for any $t\geq 0$, showing that the dynamics in \eqref{CMV} are continuous.
\end{proof}

The proof of Theorem \ref{thm_cont} relied on the following observation. 
\begin{proposition}[Jump size constraint]\label{fp_constraint} Setting $\mathbf{L}_v(t):= \sum_{l=1}^kv_l \mathcal{L}_l(t)$, any c\`adl\`ag solution to the mean-field problem \eqref{CMV} satisfies the fixed point constraint
\begin{equation}\label{eq:fp_constraint}
\Delta\mathbf{L}_v(t) = \Xi (t; \Delta\mathbf{L}, v) \quad \text{for all} \quad v\in S(\mathbf{v}),
\end{equation}
for any $t\in[0,T]$, where $\Xi(t;f,v) :=  \sum_{l=1}^k v_l \Xi_l(t;f)$ with
\vspace{-3pt}\begin{equation}\label{the_map_Xi}
\begin{cases}
\Xi_l(t;f) :=\displaystyle  \int\!u_{l}  \mathbb{P} \bigl( X_{u,v}(t\shortminus) \in [0, \Theta(t; f, u,v ) ]  , \, t\leq \tau_{u,v}  \mid B_0\bigr) d\hat{\varpi}(u,v) \\[4pt]
\displaystyle \Theta(t;f,u,v):= \log \Bigl(  1 + \displaystyle \frac{1\!-\!R}{\Lambda_{u,v}} \!\int_0^{t\shortminus} \frac{\psi(T,s)}{\psi(T,0)} d\mathbf{L}_v(s)   +\frac{1\!-\!R}{\Lambda_{u,v}}\frac{\psi(T,s)}{\psi(T,0)} f(v)  \Big) -  F_{u,v}(t\shortminus)
\end{cases}
\end{equation}
\end{proposition}
\begin{proof}
Let $\mathbf{L}_v$ be as in the statement. Clearly, the c\`adl\`agness entails $t\mapsto  \mathbf{L}_v(t)$ is c\`adl\`ag, and the dynamics further imply that $X_{u,v}$ has a jump-discontinuity at time $t$ if and only if this is the case for $\mathbf{L}_v$. By our assumptions, $v \cdot \hat{u}\geq 0$ for all $\hat{u}\in S(\mathbf{u})$ and $v\in S(\mathbf{v})$, so we get
\begin{equation}\label{eq:DeltaL}
\Delta \mathbf{L}_v(t)= \int_{\mathbb{R}^k\times\mathbb{R}^k} v \cdot \hat{u}  \bigl( \mathbb{P}( t\leq \tau \mid B_0) - \lim_{s \uparrow t} \mathbb{P}( s\leq \tau_{\hat{u},\hat{v}} \mid B_0)\bigl)  d\hat{\varpi}(\hat{u},\hat{v})
\end{equation}
for all $v\in S(\mathbf{v})$. Furthermore, by the c\`adl\`agness, the dynamics \eqref{CMV} give that any jump satisfies $\Delta X_{\hat{u},\hat{v}}(t) = \Theta(t; \Delta \mathbf{L} , \hat{u},\hat{v}) $, where $\Theta$ is as defined in \eqref{the_map_Xi}, and hence we have
\[
\mathbb{P}( t=\tau_{\hat{u},\hat{v}} \mid B_0) = 	\mathbb{P}(X_{\hat{u},\hat{v}}(t\shortminus)\in [0, \Theta(t; \Delta \mathbf{L} , \hat{u},\hat{v})] ,\;t\leq \tau_{\hat{u},\hat{v}} \mid B_0),
\]
for all $\hat{u},\hat{v}\in S(\mathbf{u},\mathbf{v})$. Combining this with \eqref{eq:DeltaL}, we therefore arrive at
\[
\Delta \mathbf{L}_v(t) = \sum_{l=1}^k v_l \int_{\mathbb{R}^k\times\mathbb{R}^k} \hat{u} \mathbb{P}(X_{\hat{u},\hat{v}}(t\shortminus)\in [0, \Theta(t; \Delta \mathbf{L} , \hat{u},\hat{v})] ,\;t\leq \tau_{\hat{u},\hat{v}} \mid B_0)  d\hat{\varpi}(\hat{u},\hat{v})
\]
for all $v\in S(\mathbf{v})$. Noting that the right-hand side is precisely $ \Xi (t; \Delta\mathbf{L}, v) $ for $\Xi $ as defined in \eqref{the_map_Xi}, this proves the desired constraint \eqref{eq:fp_constraint}, and so we are done.
\end{proof}
The closest to \eqref{CMV} in the literature is the McKean--Vlasov problem studied in \cite{NS18}, whose `multitype' network structure is a special case of the framework introduced here.
We emphasise, however, that our treatment differs markedly from \cite{NS18}, by focusing on the financial underpinnings and connecting the mean-field problem to a finite particle system.  Moreover, we note that \cite{NS18} is focused on criteria for the dynamics to undergo a jump, while we obtain the above criterion for continuity. Finally, concerning jumps, our cascade condition sheds new light on \cite{NS18} in terms of how to characterise jumps in the mean-field problem, as we discuss below. Note, however, that \cite{NS18} also introduces a game version, where the strength of interactions between types are chosen strategically, something we do not consider here.

In the absence of \eqref{smallness_cond}, it is less clear if \eqref{CMV} is well-posed, as zero may no longer be the only solution to the jump size constraint \eqref{eq:fp_constraint} and there could be multiple non-zero solutions. Thus, we need to identify which jump size satisfying \eqref{eq:fp_constraint} is selected in the limit by the cascade condition \eqref{particle_cascade_cond1}. With $\Xi$ and $\Xi_l$ from \eqref{the_map_Xi}, we conjecture that \eqref{particle_cascade_cond1} selects solutions to \eqref{CMV} with jumps $\Delta\mathcal{L}_l=\Xi_l(t;\Delta\mathbf{L})$ given by the following \emph{mean-field cascade condition}
\begin{equation}\label{eq:limit_PJC}
\Delta \mathbf{L}_v(t) = \lim_{\varepsilon\downarrow 0} \lim_{m \uparrow \infty}  \Delta^{\!(m,\varepsilon)}_{t,v}, \quad \Delta^{\!(m,\varepsilon)}_{t,v}= \Xi(t; \varepsilon +  \Delta^{\!(m-1,\varepsilon)}_{t,\hspace{0.5pt}\cdot}, v), \;\; \Delta^{\!(0,\varepsilon)}_{t,v}= \Xi(t;\varepsilon , v).
\end{equation}
Unlike in \eqref{particle_cascade_cond1}, $\Xi(t;0,\cdot)$ is always zero, which explains the need for an $\varepsilon$ perturbation. Note also that \eqref{eq:limit_PJC} is indeed well-defined, since $\Delta^{(m,\varepsilon)}_{t,v}$ forms a bounded sequence increasing as $m\uparrow\infty$ and decreasing as $\varepsilon\downarrow0$; crucially, dominated convergence shows that $ \Delta \mathbf{L}$ given by \eqref{eq:limit_PJC} satisfies the jump size constraint \eqref{eq:fp_constraint}. The iterative structure of \eqref{eq:limit_PJC} is a particularly nice feature making it easy to implement numerically, as is completed in Supplemental~\ref{sec:4example} -- see Figure \ref{fig:heat_plots} -- to illustrate \eqref{CMV} and \eqref{eq:limit_PJC}. A further discussion of the mean-field cascade condition will form part of \cite{fein-soj}.

\section{Conclusion}\label{sec:conclusion}

Our main contribution is two-fold.  First, we proposed a novel dynamic default contagion model with early defaults driven by insolvency that generalises the Gai--Kapadia framework; a comparable illiquidity model is presented in the Online Supplemental. Second, starting from our finite bank setting, we proposed a tractable framework for reformulating the model as a stochastic particle system with a limiting mean-field problem for which we developed a constraint on the interactions that guarantees continuous evolution of the system.  

We wish to present two important extensions of our model, which we feel deserve further study: (i) In this work, we assumed no banks rebalance their asset holdings over time; in particular, we assume all incoming interbank payments are held in the risk-free asset until the final maturity time.  This optimal reinvestment and rebalancing question as an optimal control problem which encodes the tradeoffs between short-term liquidity and long-term equity -- particularly in the setting of joint illiquidity and insolvency defaults presented in Supplemental~\ref{sec:illiquid} -- is of great importance.  (ii) The Gai--Kapadia model which we consider presents a simplified balance sheet structure.  More granular balance sheet information so as to encode, e.g., margin calls and fire sales (as studied in, e.g.~\cite{bichuch2019optimization,cont2020liquidity}), could greatly alter the default times and cascades.  This problem is particularly tied with the liquidity of assets and thus, also, to the illiquidity problem presented in Supplemental~\ref{sec:illiquid}.

\bibliographystyle{plain}
\bibliography{bibtex2}

\newpage
\appendix
\section{Proof of Proposition \ref{prop:particle}}\label{proof:prop:particle}
In this section we give the proof of Proposition \ref{prop:particle}, showing that the stochastic particle system \eqref{particle_sys}-\eqref{particle_cascade_cond1} corresponds uniquely to the greatest clearing capital from Proposition \ref{prop:exist}.
\begin{proof}[Proof of Proposition \ref{prop:particle}]
We begin by verifying that any clearing capital  $K$ satisfying \eqref {eq:Kapital} gives a solution to \eqref{particle_sys} after applying the distance-to-default transformation \eqref{dist-to-def}, and then we show that the cascade condition \eqref{particle_cascade_cond1} for the resolution of instantaneous default cascades indeed corresponds to selecting the greatest clearing capital $K^\uparrow$ from Proposition \ref{prop:exist}. Notice that, given the dynamics \eqref{GBM} for the external assets, we can compute explicitly the conditional expectations $E[x_i(T) \; | \; \fcal_t]=x_i(t) e^{\int_t^T \!\mu_i(s) ds}$. Using also that $\lambda_{ij}=u^i\cdot v^j$ for $i\neq j$, by \eqref{lambda_fact}, we thus deduce that any solution $K$ to \eqref{eq:K_i} can be rewritten as a coupled system\vspace{-3pt}
\begin{equation}\label{1st_K_i}
\begin{cases}
K_i(t) =\displaystyle x_i(t) e^{\int_t^T \! \mu_i (s)ds  } - \psi(T,0)\Lambda_i  - \int_0^t (1-R)\psi(T,s)dL^n_i(s)  \\
L^n_i(t):= \displaystyle \sum_{j=1}^{n} \lambda_{ji} \mathbf{1}_{t\geq \tau_j} = \sum_{l=1}^k v^i_l  \sum_{j\neq i} u^j_l\mathbf{1}_{t\geq \tau_j}, \;\;\;\tau_i=\inf\{ t\geq0 \; | \; K_i(t) \leq 0  \},
\end{cases}
\end{equation}
for $i=1,\ldots,n$, where we note that the definitions of $\mathbf{L}_{v}^n$ and $\mathcal{L}_{l,i}^n$ in \eqref{particle_sys}-\eqref{particle_cascade_cond1} imply
\begin{equation}\label{eq:fat-L}
\mathbf{L}_{v^i}^n(t)=\sum_{l=1}^kv_l(\mathcal{L}^n_{l,i}(t)+u_l^i\mathbf{1}_{t\geq \tau^i}) = L_i^n(t) + u^i \cdot v^i \mathbf{1}_{t \geq \tau_i}.
\end{equation}
Applying the transformation \eqref{dist-to-def} to the expression for $K_i$ in \eqref{1st_K_i}, we obtain the dynamics \eqref{particle_sys} for $X_i$ as defined in \eqref{dist-to-def}, where the equality of the stopping times follows from the fact that $K_i(t)\leq 0$ if and only if $X_i(t)\leq 0$, for each $i\in \mathcal{N}$. Of course, different solutions $K$ correspond to different stopping times.

From here onwards, we let $X_i$ (and hence the stopping times $\tau_i$) be defined from $K^\uparrow$, and we then need to check that $X_i$ indeed satisfies \eqref{particle_sys}-\eqref{particle_cascade_cond1}, where we stress that $K^\uparrow$ is c\`adl\`ag, by the second part of Proposition \ref{prop:exist}. As there are $n$ banks, there can be at most $n$ defaults, and hence the system can undergo at most $n$ jumps. At time $t=0$, the initial condition is such that all banks have strictly positive capital, so we can let $\varsigma_1>0$ denote the first (random) time that $K^\uparrow_{i}(t\shortminus)\leq0$ for some ${i}\in\mathcal{N}$. For $t\in [0,\varsigma_1)$ we have $\mathcal{D}_t=\emptyset$ and $\Delta \mathcal{L}^{n}_{l,i}(t)=0$ for all $l$ and $i$, since $K^\uparrow_i>0$ on $[0,\varsigma_1)$ for all $i$ gives $\Delta_{t,\cdot}^{n,(0)}=0$ and hence $\Delta \mathbf{L}^n(t)=0$. Thus, each $X_i$ satisfies \eqref{particle_sys}-\eqref{particle_cascade_cond1} on $[0,\varsigma_1)$. Now let
\begin{align*}
\mathcal{D}^{(0)}_{\varsigma_1} := \{i \; | \;  X_i(\varsigma_1 \shortminus) \leq \Theta^n(\varsigma_1;0;i) \}= \{i \; | \;  X_i(\varsigma_1 \shortminus) \leq 0 \}  = \{i \; | \; K^\uparrow_i(\varsigma_1\shortminus)=0\}\neq \emptyset,
\end{align*}
and note that
\begin{equation}\label{eq:1st_jump}
\sum_{l=1}^k v_{l}^i \sum_{j\in\mathcal{D}_{\varsigma_1}^{(0)} }u_{l }^j = \Xi^n (\varsigma_1; 0, v_i ) = \Delta^{n,(0)}_{\varsigma_1,v^i} , \qquad \text{for} \quad i=1,\ldots,n,\vspace{-4pt}
\end{equation}
Since $	\mathcal{D}^{(0)}_{\varsigma_1} \subseteq \{i \; | \; \tau_i = \varsigma_1\}$, it follows that the loss processes $\mathbf{L}_i^n$ in \eqref{1st_K_i} must jump up by at least the amount \eqref{eq:1st_jump} at time $\varsigma_1$, and hence, even under the greatest clearing capital, any banks $i$ with $X_i(\varsigma_1 \shortminus) \leq \Theta^n(\varsigma_1;\Delta^{n,(0)}_{\varsigma_1,\cdot};i) $ must necessarily be shifted to non-positive capital $K^\uparrow_i(\varsigma_1)\leq0$ precisely at time $\varsigma_1$. That is,
\[
\mathcal{D}^{(1)}_{\varsigma_1}:= \{ i \; | \;  X_i(\varsigma_1 \shortminus) \leq \Theta^n(\varsigma_1;\Delta^{n,(0)}_{\varsigma_1,\cdot};i)   \} \subseteq	\{i \; | \; \tau_i=\varsigma_1\}.
\]
From here, we can then observe that
\[
\sum_{l=1}^k v_{l}^i \sum_{j\in\mathcal{D}_{\varsigma_1}^{(1)} }u_{l }^j = \Xi^n (\varsigma_1; \Delta^{n,(0)}_{\varsigma_1,\cdot}, v_i ) = \Delta^{n,(1)}_{\varsigma_1,v^i} , \qquad \text{for} \quad i=1,\ldots,n,\vspace{-4pt}
\]
so, arguing precisely as above, it follows from $\mathcal{D}^{(1)}_{\varsigma_1} \subseteq \{i \; | \; \tau_i=\varsigma_1\}$ that we must have
\[
\mathcal{D}^{(2)}_{\varsigma_1}:= \{ i \; | \;  X_i(\varsigma_1 \shortminus) \leq \Theta^n(\varsigma_1;\Delta^{n,(1)}_{\varsigma_1,\cdot};i)   \} \subseteq 	\{i \; | \; \tau_i=\varsigma_1\}.
\]
Continuing in this way, we further obtain that, for any $m\in\{0,\ldots,n\}$,
\begin{equation}\label{eq:default_set_iterated}
\mathcal{D}^{(m)}_{\varsigma_1}:=  \{ i \; | \;  X_i(\varsigma_1 \shortminus) \leq \Theta^n(\varsigma_1;\Delta^{n,(m-1)}_{\varsigma_1,\cdot};i)   \} \subseteq 	\{i \; | \; \tau_i=\varsigma_1\}.
\end{equation}
Since there can be at most $n$ defaults, there is an $\bar{m}\leq n-1$ such that  $\Delta^{n,(m)}_{\varsigma_1,\cdot}=\Delta^{n,(\bar{m})}_{\varsigma_1,\cdot}$ and $\mathcal{D}^{(m)}_{\varsigma_1}=\mathcal{D}^{(\bar{m})}_{\varsigma_1}$ for $m=\bar{m},\ldots,n$. From this, we deduce that
\[
\mathcal{D}^{(\bar{m})}_{\varsigma_1}=\{ i\; | \; K_i(\varsigma_1\shortminus)\leq (1-R)\psi(T,\varsigma_1) \Delta^{n,(\bar{m})}_{\varsigma_1,v^i} \},
\]
so, taking $\bar{\Delta}_i^n:=\Delta^{n,(\bar{m})}_{\varsigma_1,v^i} - u^i\cdot v^i \mathbf{1}_{i\in \mathcal{D}^{(\bar{m})}_{\varsigma_1}}$ for each $i=1,\ldots,n$, we get
\[
\bar{\Delta}_i^n = \sum_{l=1}^k v_{l}^i \sum_{j\neq i}^nu_{l }^j\mathbf{1}_{\{K_j(\varsigma_1\shortminus)\leq (1-R)\psi(T,\varsigma_1) \bar{\Delta}_j^n \}}.
\]
This means that letting each $L_i^n(t)$ have a jump of size $\bar{\Delta}_i^n$ at time $\varsigma_1$ with corresponding default set $\mathcal{D}^{(\bar{m})}_{\varsigma_1}$ yields a solution to \eqref{1st_K_i} on $[0,\varsigma_1]$. If $\mathcal{D}^{(\bar{m})}_{\varsigma_1}$ is a strict subset of $\{i \; | \; \tau_i=\varsigma_1\}$, then we have found a solution to \eqref{1st_K_i} with greater clearing capital than $K^\uparrow$ at time $\varsigma_1$, which is of course a contradiction, so we conclude from this and \eqref{eq:default_set_iterated} that
\[
\mathcal{D}_{\varsigma_1}:=\{i \; | \; \tau_i=\varsigma_1\}=\mathcal{D}^{(\bar{m})}_{\varsigma_1}=\{ i \; | \;  X_i(\varsigma_1 \shortminus) \leq \Theta^n(\varsigma_1;\Delta^{n,(\bar{m})}_{\varsigma_1,\cdot};i)   \}.
\]
with jump sizes
\[\Delta L^n_i(t) = \Delta^{n,(\bar{m})}_{\varsigma_1,v^i} - u^i\cdot v^i \mathbf{1}_{i\in \mathcal{D}_{\varsigma_1}}.
\]
Recalling \eqref{eq:fat-L}, this indeed implies
\[
\Delta \loss^n_{v^i}(t) =\Delta^{n,(\bar{m})}_{\varsigma_1,v^i} = \lim_{m\rightarrow n} \Delta^{n,(m)}_{\varsigma_1,v^i},
\]
and hence also $\Delta \mathcal{L}^{n}_{l,i}(t) = \Xi^n_l(t,\Delta \mathbf{L}^n) - u_l^i\mathbf{1}_{i\in \mathcal{D}_t}$, by definition of $\Xi_{l}^n$. In conclusion, $X_i$ from \eqref{dist-to-def} defined in terms of $K^\uparrow$ is the unique solution to \eqref{particle_sys}-\eqref{particle_cascade_cond1} on $[0,\varsigma]$.

It remains to consider the $k^{th}$ (random) time $\varsigma_k>\varsigma_{k-1}$ that we have $K^\uparrow(\varsigma_k\shortminus)\leq0$, for a general $k=2,\ldots,n$. For such a $k$, we simply define
\[
\mathcal{D}^{(m)}_{\varsigma_k}:=  \{ i \; | \;  X_i(\varsigma_k \shortminus) \leq \Theta^n(\varsigma_k;\Delta^{n,(m-1)}_{\varsigma_1,\cdot};i),\; \varsigma_k\leq \tau_i   \}
\]
in place of \eqref{eq:default_set_iterated}, noting that we have included the requirement $\varsigma_k\leq \tau_i $, so that we are not counting banks that have already defaulted strictly before time $\varsigma_k$. From here, we can argue precisely as above to conclude that $X^i$ defined from $K^\uparrow$ is the unique solution to \eqref{particle_sys}-\eqref{particle_cascade_cond1} on $[0,\varsigma_k]$, and hence the proof is complete, since the system has standard diffusive dynamics for the remaining time after the last default.
\end{proof}

\section{Multi-type mean-field network example}\label{sec:4example}
To illustrate the tiered grouping structure presented in Section~\ref{sec:mean-field}, setting $k=4$ gives a simple way of splitting the system into four subgroups, where, e.g.,~$u^i=(0,u^i_2,0,0)$ says that bank $i$ belongs to the second grouping of top-tier banks with the score $u^i_2$ giving its relative importance, while $u^j=(0,0,u^j_3,0)$ says that bank $j$ belongs to the first grouping of lower-tier banks with relative importance $u^j_3$, and so on. The tiered structure is then obtained by letting $v^1,v^2$ spread the top-tier lending with large values in the first two entries (lending between the top-tier) and lower values in the last two entries (less lending to the lower-tier), while $v^3,v^4$ has small values in the first two entries (some lending to the top-tier) and zeroes in the last two entries (no lending between the lower-tier). See Figure \ref{network}.
\begin{figure}[H]
	\centering
	\begin{minipage}[b]{0.8\linewidth}
		\includegraphics[width=\textwidth]{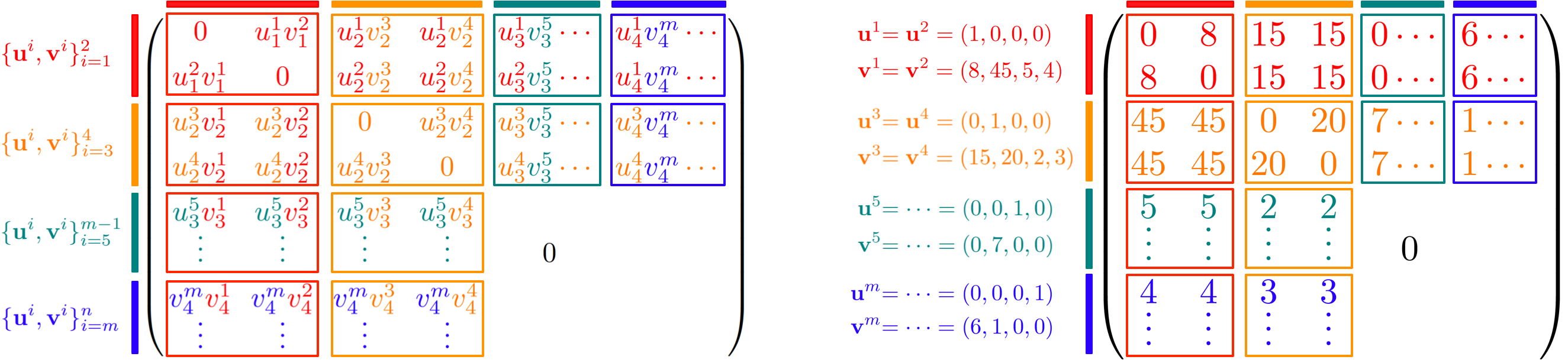}
	\end{minipage}
	\caption{The first picture shows the general structure of the relative liabilities $(\lambda_{ij})$ for the simple $k=4$ example described above. The second picture gives a concrete realisation of this with \emph{no} heterogeneity within groups.}\label{network}
\end{figure}
Figure \ref{limit_meas} illustrates the limiting mean-field model of this example. 
\begin{figure}[H]\vspace{-3pt}
\centering
\begin{minipage}[b]{0.75\linewidth}
	\includegraphics[width=\textwidth]{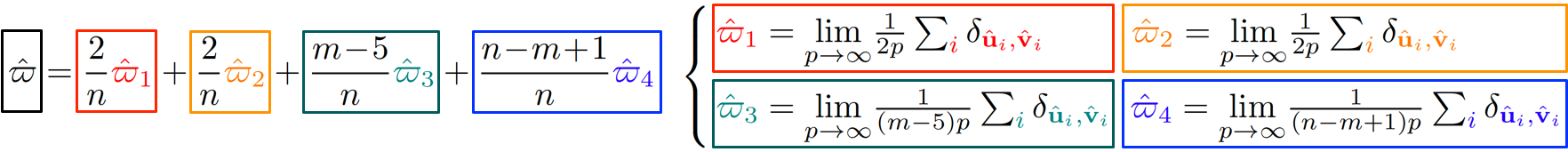}
\end{minipage}\vspace{-2pt}
\caption{Limiting $\hat{\varpi}$ for the particular four-group network structure in Figure \ref{network}. In the simplest case of homogeneity within groups, the $\hat{\varpi}_i$'s are simply given by the point masses $\delta_{{\color{red}u^1,v^1}}$, $\delta_{{\color{orange}u^3,v^3}}$, $\delta_{{\color{teal}u^5,v^5}}$, and $\delta_{{\color{blue}u^m,v^m}}$.}\label{limit_meas}.
\end{figure}\vspace{-10pt}

With this concrete example in mind, consider the `four-homogeneous-types' network structure given by the rightmost liabilities matrix in Figure \ref{network}, where we take $n=12$ with $2$ banks of each of the two top-tier types (red and orange), and $4$ of each of the two lower-tier types (green and blue). The corresponding mean-field network, described by Figure \ref{limit_meas}, is then modelled by 
\begin{equation}\label{eq:example_cont}
(\mathbf{u},\mathbf{v} ) \sim \hat{\varpi}= \frac{1}{6}\delta_{({\color{red}e_1}, {\color{red}(8,45,5,4)})} + \frac{1}{6} \delta_{({\color{orange}e_2}, {\color{orange}(15,20,2,3)})} +\frac{1}{3} \delta_{({\color{teal}e_3}, {\color{teal}7e_2})}  + \frac{1}{3} \delta_{({\color{blue}e_4}, {\color{blue}(6,1,0,0)})},
\end{equation}
and we have four initial conditions $V^i_0(\cdot)=V_0(\cdot|u^i,v^i)$ for $i=1,\ldots,4$. For this simple example, the condition \eqref{smallness_cond} just reads as four inequalities
\begin{equation}\label{eq:continuity_4_types}
\Vert V^1_0 \Vert_\infty < \frac{\Lambda_1}{1-R} \frac{1}{15}, \quad 	\Vert V^2_0 \Vert_\infty < \frac{\Lambda_2}{1-R} \frac{1}{45}, \quad \Vert V^3_0 \Vert_\infty < \frac{\Lambda_3}{1-R} \frac{1}{5},\quad\Vert V^4_0 \Vert_\infty < \frac{\Lambda_4}{1-R} \frac{1}{4} .
\end{equation}
Notice from the expression for $X_{u,v}(0)$ that if $v_0^i$ denotes the density of the type $i$ initial external asset values, then we have $V_{0}^i(x)=v^i_0(\Lambda_{i} \psi(T,0)e^{-\!\int_0^T\!\!\mu_{i}(s)ds}e^x)\Lambda_{i} \psi(T,0)e^{-\!\int_0^T\!\!\mu_{i}(s)ds}e^x$.

\begin{figure}[H]
	\centering
	\begin{minipage}[b]{0.23\linewidth}
		\includegraphics[width=\textwidth]{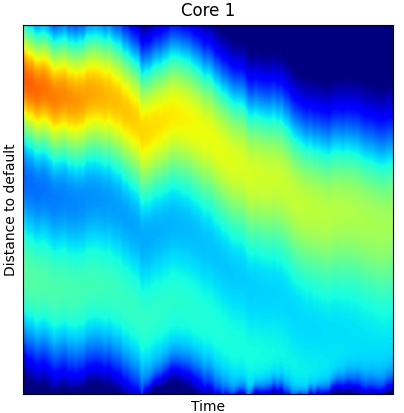}
\end{minipage}
\begin{minipage}[b]{0.23\linewidth}
\includegraphics[width=\textwidth]{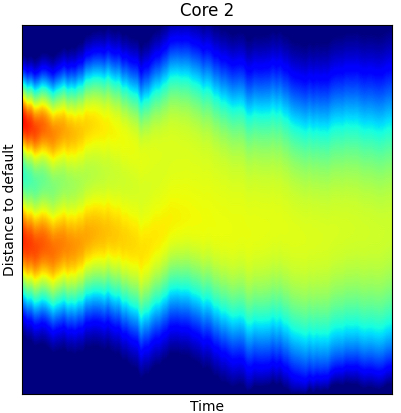}
\end{minipage}
\begin{minipage}[b]{0.23\linewidth}
\includegraphics[width=\textwidth]{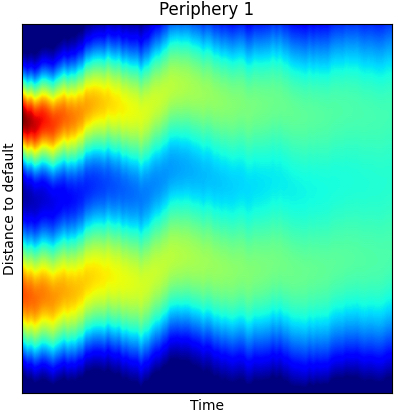}
\end{minipage}
\begin{minipage}[b]{0.23\linewidth}
\includegraphics[width=\textwidth]{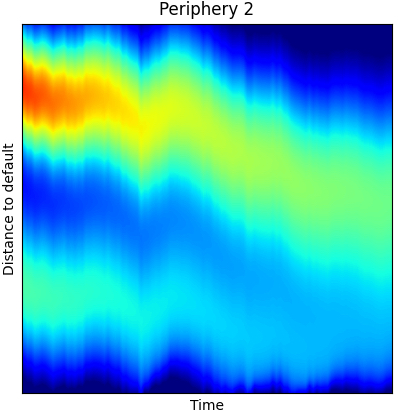}
	\end{minipage}
	\caption{Density heat plots of \eqref{CMV}, for a given realisation of $B_0$, with the network structure \eqref{eq:example_cont}. The bimodal initial densities are sufficiently spread out to satisfy \ref{eq:continuity_4_types}, ensuring continuity. Banks in `Core 1' and `Periphery 2' are exposed to each other, which creates a sharp decline due to insolvency contagion from their lower performing fractions. During the decline, $B_0$ swings up, thus stopping the agony temporarily, but a downward trend returns, and this is amplified by continual contagion in the lower performing fractions. `Core 2' feels some effect of this, while `Periphery 1' is largely untouched. Compare with the much more drastic turn of events in Figure \ref{fig:heat_plots} below.}
	\label{fig:heat_plots2}
\end{figure}

\begin{figure}[H]
\centering
\begin{minipage}[b]{0.23\linewidth}
	\includegraphics[width=\textwidth]{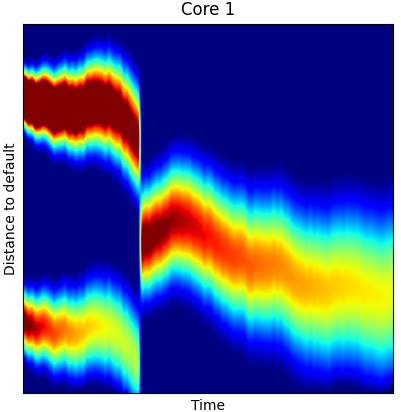}
\end{minipage}
\begin{minipage}[b]{0.23\linewidth}
	\includegraphics[width=\textwidth]{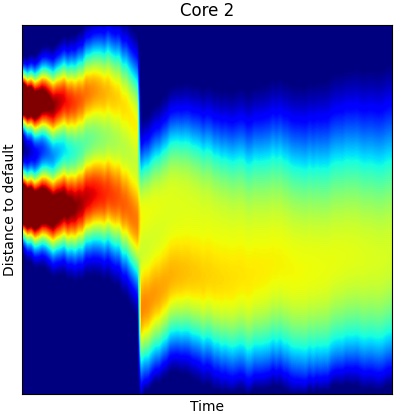}
\end{minipage}
\begin{minipage}[b]{0.23\linewidth}
	\includegraphics[width=\textwidth]{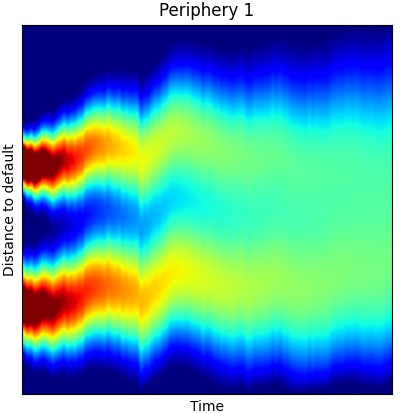}
\end{minipage}
\begin{minipage}[b]{0.23\linewidth}
	\includegraphics[width=\textwidth]{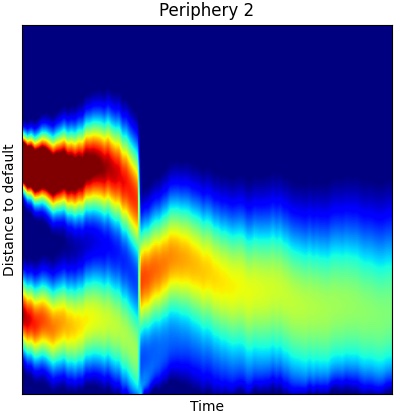}
\end{minipage}
\caption{Density heat plots of \eqref{CMV} for the same setting as Figure \ref{fig:heat_plots2}, only with different initial conditions that are now highly concentrated at each mode, violating \eqref{smallness_cond}. The plots show a jump, which occurs as an instant default cascade between the low performing fractions of `Core 1' and `Periphery 2'. This spills into a severe downgrading of the health of `Core 2', but only a negligible proportion of it defaults, as `Core 2' was otherwise performing well. Since `Periphery 1' is only exposed to defaults in `Core 2', the events have little impact on `Periphery 1'.}
\label{fig:heat_plots}
\end{figure}

\section{Illiquidity and insolvency}\label{sec:illiquid}
Within the main part of the text, defaults occur due to insolvency.  In this section, we consider defaults due to illiquidity.  First, we will focus on the dynamics of illiquid defaults \emph{only} and then joint condition on liquidity and capital.  For simplicity of exposition and description we will focus solely on the finite system with parameters as described in Section~\ref{sec:dynamic}.  Herein, the interpretation of the external asset $x(t)$ at intermediate time $t$ is important; we will take $x(t)$ to denote the amount of (external) cash available to each bank at time $t$ which can be used to pay debts.  As noted in Section~\ref{sec:dynamic}, the cash available $x(t)$ may differ from the expected future assets $\bbe[x_i(T) \; | \; \fcal_t]$.

\paragraph{A.1 Illiquidity only}
\begin{figure}[t]
\centering
\begin{subfigure}[t]{0.42\textwidth}
\centering
\begin{tikzpicture}[x=\linewidth/4.9,y=5mm]
\draw[draw=none] (0,9.5) rectangle (5,10) node[pos=.5,yshift=0.3em]{};
\draw[draw=none] (0,9) rectangle (3,9.5) node[pos=.5,yshift=0.2em]{\small \bf Assets};
\draw[draw=none] (3,9) rectangle (5,9.5) node[pos=.5,yshift=0.2em]{\small \bf Liabilities};

\filldraw[fill=blue!20!white,draw=black] (0,6.5) rectangle (3,9) node[pos=.5,style={align=center}]{\footnotesize External \\ ${\scriptstyle x_i(t)}$};
\filldraw[fill=yellow!20!white,draw=black] (0,4) rectangle (3,6.5) node[pos=.5,style={align=center}]{\footnotesize Interbank (Solvent) \\ ${\scriptstyle\sum_{j \in \acal_t} L_{ji}(t)}$};
\filldraw[fill=orange!20!white,draw=black] (0,0) rectangle (3,4) node[pos=.5,style={align=center}]{\footnotesize Interbank (Insolvent) \\ ${\scriptstyle\sum_{j \in \ncal \backslash \acal_t} \Bigl(\begin{array}{l}{\scriptstyle (1 - R) L_{ji}(\tau_j)}\\ {\scriptstyle \;+ \,R L_{ji}(T)}\end{array}\Bigr)}$};
			
\filldraw[fill=red!20!white,draw=black] (3,3) rectangle (5,9) node[pos=.5,style={align=center}]{\footnotesize Total \\ ${\scriptstyle\sum_{j \in \ncal_0} L_{ij}(t)}$};
\filldraw[fill=green!20!white,draw=black] (3,0) rectangle (5,3) node[pos=.5,style={align=center}] (t) {\footnotesize Cash Account \\ ${\scriptstyle V_i(t)}$};
\end{tikzpicture}
\caption{\footnotesize Stylised cash flow statement for bank $i$ at time $t$.}
\label{fig:cash-flow}
\end{subfigure}
~
\begin{subfigure}[t]{0.47\textwidth}
\centering
\includegraphics[width=0.73\linewidth]{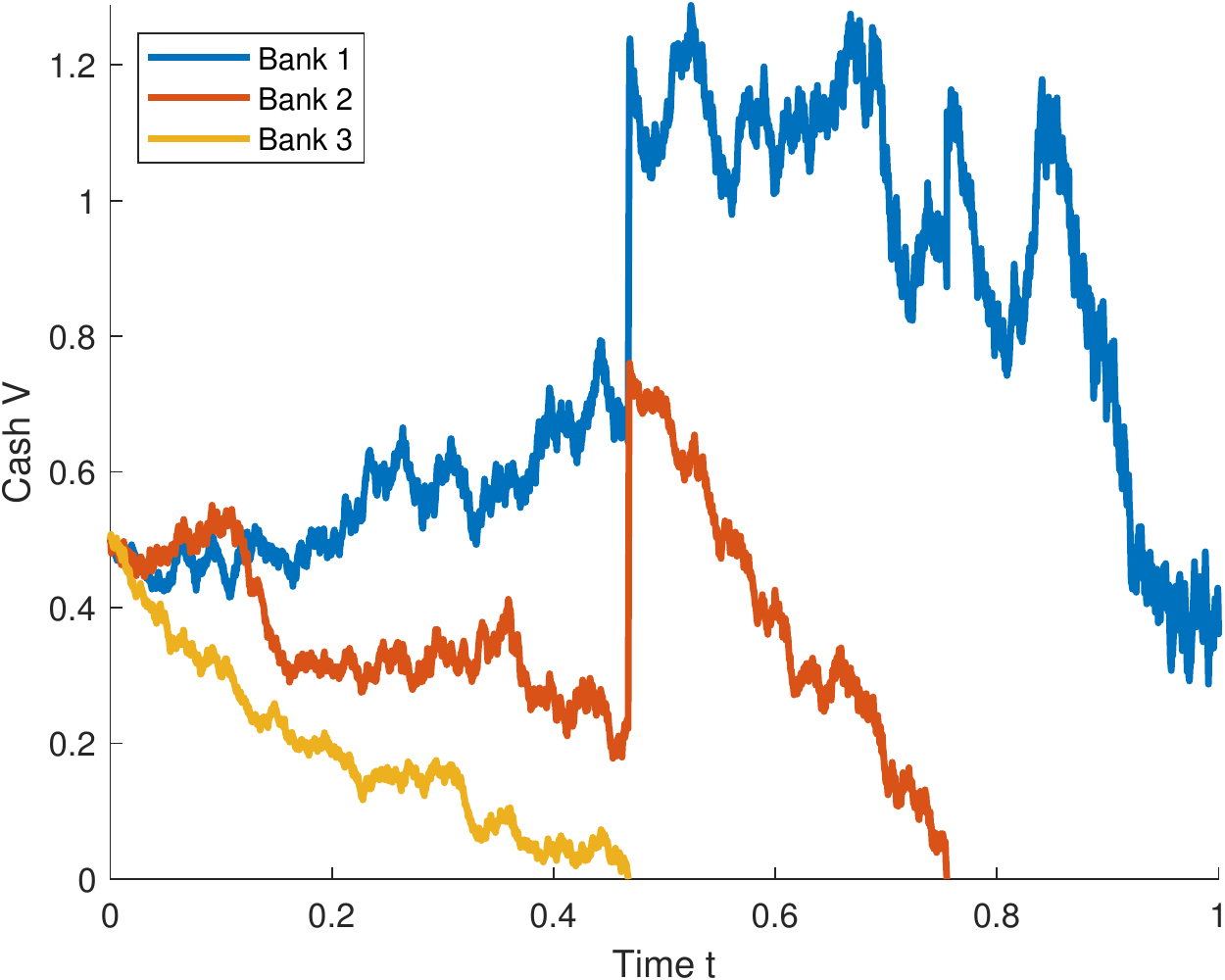}
\caption{\footnotesize Realisation of a 3 bank system driven by correlated GBMs.}
\label{fig:cash-3bank}\vspace{-4pt}
\end{subfigure}
\caption{Stylised cash flow statement at time $t$ and a realisation of this system over time.}\vspace{-8pt}
\end{figure}
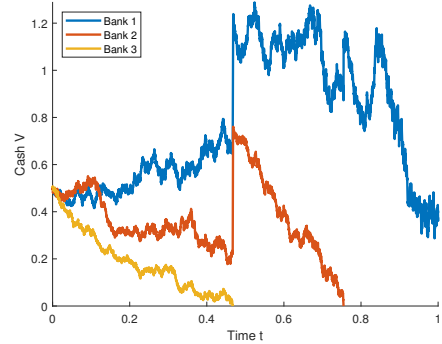
Akin the stylised balance sheet displayed in Figure~\ref{fig:balance-sheet}, the cash flow statement for bank $i$ is provided in Figure~\ref{fig:cash-flow}.  As opposed to the main text, herein we wish to consider defaults caused by illiquidity, i.e., when the cash account drops below 0.  In order to consider the cash account, we need to determine the impacts of a default on the cash account of the defaulting firm's counterparties.  Taking the notion from, e.g.,~\cite{KV16}, the \emph{future} obligations become due at the default time; as in the main text, we assume that these obligations have a fixed recovery rate $R$ which is paid to the defaulting firm's counterparties.  Given these simplifying financial assumptions, the cash account evolves according to:
\begin{align}
\label{eq:liquidity} V_i(t) &= x_i(t) + \sum_{j \in \acal_t} L_{ji}(t) + \sum_{j \not\in \acal_t}\left[(1-R)L_{ji}(\tau_j) + R L_{ji}(T)\right] - \sum_{j \in \ncal_0} L_{ij}(t)\\
\label{eq:illiquid-time} \tau_i &= \inf\{t \in [0,T] \; | \; V_i(t) \leq 0\}\\
\label{eq:liquid-set} \acal_t &= \{i \in \ncal \; | \; \tau_i > t\}
\end{align}
for every bank $i$ and time $t$.

Notably, when a default occurs ($V_i(t) \leq 0$) then the liquidity of all solvent firms actually jumps \emph{up} because the future unpaid liabilities are paid off early (at the default time) as is studied in, e.g., \cite{KV16}, because $(1-R) L_{ji}(t) + R L_{ji}(T) \geq L_{ji}(t)$ for any time $t$. However, with this jump in the short term cash account, the long-term trend for the cash account will experience an increased downward drift (or a lesser upward drift) after the default time since interbank assets $L_{ji}$ are now cut off at the default time.  This implies that this illiquidity setting in continuous time can\emph{not} cause any contagious defaults.  This differs from the discrete time setting of~\cite{KV16} due to continuity assumptions of this model.  This insight leads to the following proposition on the uniqueness of the clearing solution to this liquidity problem.
\begin{proposition}[Clearing cash account]\label{prop:exist-cash}
(Lebesgue) almost surely, there exists a unique clearing cash account $V^*$ to the network clearing problem defined by~\eqref{eq:liquidity}. Furthermore, this clearing cash account is c\`adl\`ag.
\end{proposition}
\begin{proof}
We will prove this result constructively.  Define
    \[V_i^{(0)}(t) := x_i(t) + \sum_{j \in \ncal} L_{ji}(t) - \sum_{j \in \ncal_0} L_{ij}(t)\]
for every bank $i$ and time $t$.
Up until the first default time $\tau_{[1]} = \inf\{t \in [0,T] \; | \; \min_{i \in \ncal} V_i^{(0)}(t) \leq 0\}$, all possible clearing solutions $V$ must coincide with $V^{(0)}$.  Let $\iota_1$ denote this first defaulting bank and there is no other bank $i \neq \iota_1(\omega)$ such that $V_i^{(0)}(\tau_{[1]},\omega) \leq 0$ for almost any $\omega \in \Omega$.   Define $\acal_t^{(0)} = \ncal$ for $t < \tau_{[1]}$ and $\acal_t^{(0)} = \ncal\backslash\{\iota_1\}$ for $t \geq \tau_{[1]}$.
Now, define
    \[V_i^{(1)}(t) := x_i(t) + \sum_{j \in \acal_t^{(0)}} L_{ji}(t) + \sum_{j \not\in \acal_t^{(0)}}\left[(1-R)L_{ji}(\tau_j) + R L_{ji}(T)\right] - \sum_{j \in \ncal_0} L_{ij}(t)\]
for every bank $i$ and time $t$.  Up until the \emph{second} default time $\tau_{[2]} = \inf\{t \in [0,T] \; | \; \min_{i > 1} V_i^{(1)}(t) \leq 0\}$, all possible clearing solutions must coincide with $V^{(1)}$ due to the upward jump in the cash accounts of solvent firms at $\tau_{[1]}$ precluding any contagious defaults at that time.  In this way, we can iteratively construct $V^{(k+1)}$ from $V^{(k)}$, i.e.,
    \[V_i^{(k+1)}(t) := x_i(t) + \sum_{j \in \acal_t^{(k)}} L_{ji}(t) + \sum_{j \not\in \acal_t^{(k)}}\left[(1-R)L_{ji}(\tau_j) + R L_{ji}(T)\right] - \sum_{j \in \ncal_0} L_{ij}(t)\]
for every bank $i$ and time $t$.  This process is repeated at each new defaulting time until (almost surely) no new defaults occur (either due to liquidity for the full period $[0,T]$ or due to early default); as with the fictitious default algorithm, such a process would require at most $n$ iterations.  Since $V^{(k)}(t)$ is c\`adl\`ag and must coincide with every clearing solution on $t \in [0,\tau_{[k]}]$, the result is proven.
\end{proof}

\begin{figure}[t]
\centering
\begin{subfigure}[t]{0.45\textwidth}
\centering
\includegraphics[width=0.73\linewidth]{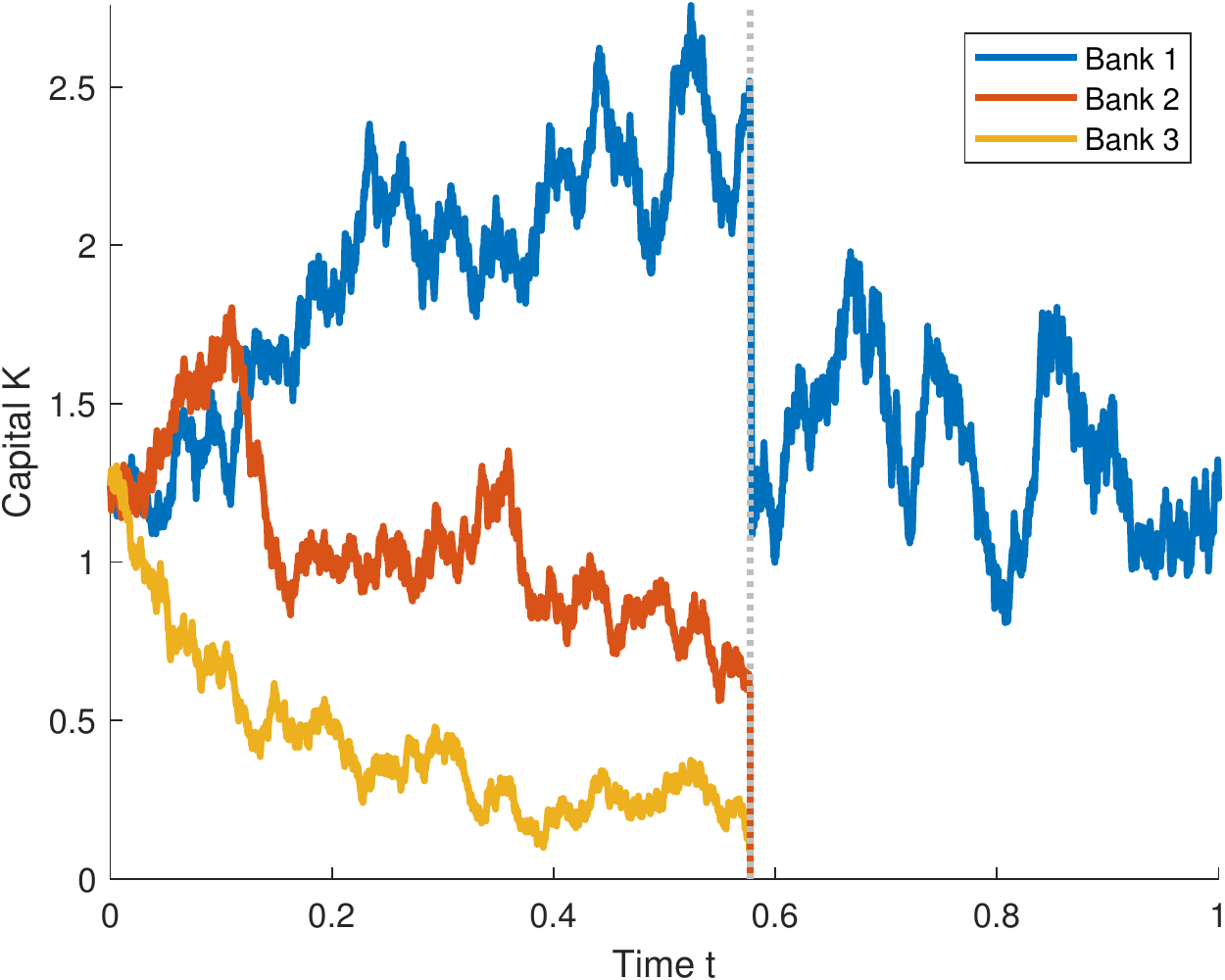}
\caption{\footnotesize Realisation of the capital account $K$.}
\label{fig:capital-joint}
\end{subfigure}
~
\begin{subfigure}[t]{0.45\textwidth}
\centering
\includegraphics[width=0.73\linewidth]{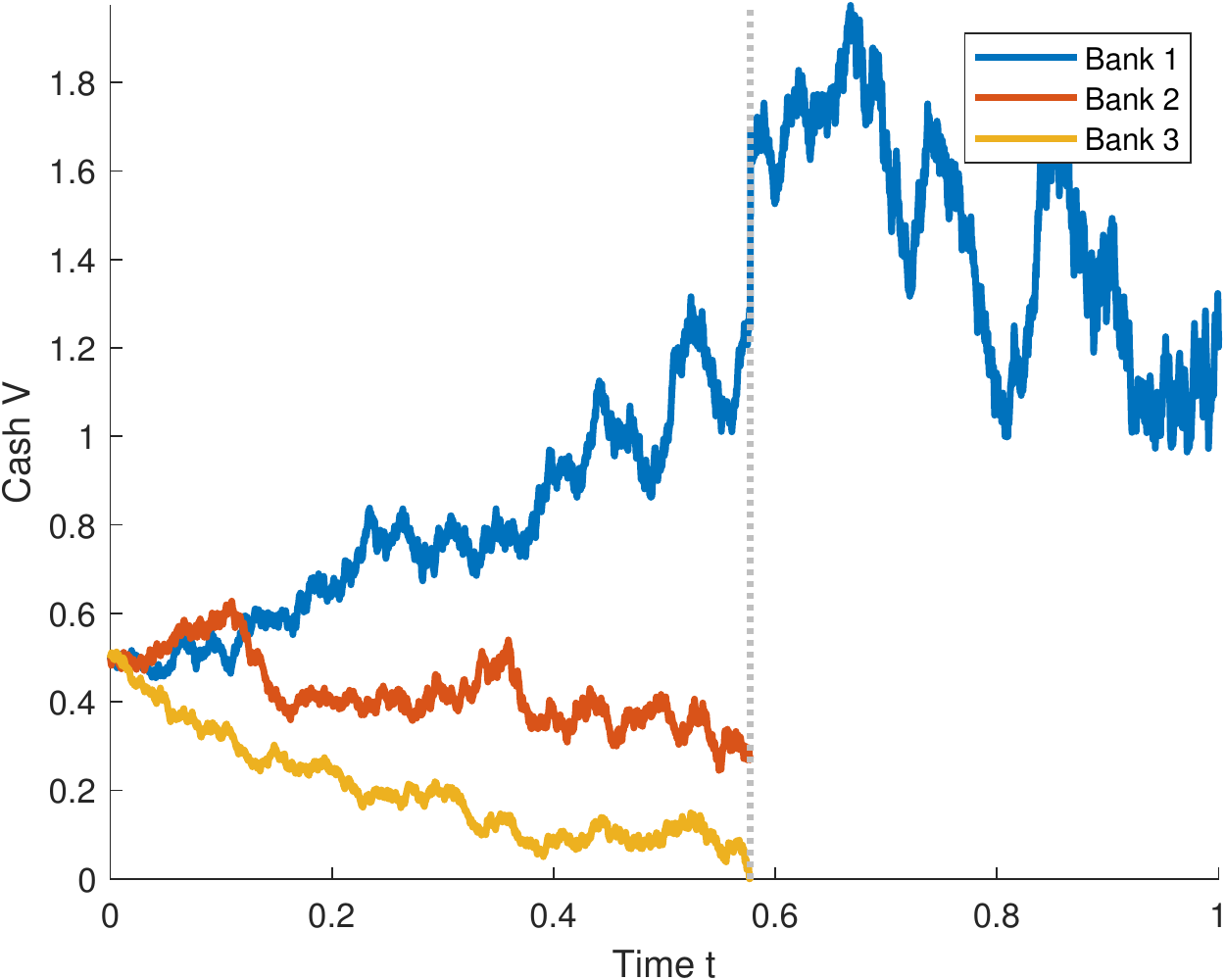}
\caption{\footnotesize Realisation of the cash account $V$.}
\label{fig:cash-joint}\vspace{-4pt}
\end{subfigure}
\caption{Realisation of joint insolvency and illiquidity system over time for a 3 bank system driven by correlated GBMs. In this system, illiquidity of bank 3 (shown in figure~\eqref{fig:cash-joint}) triggers the insolvency of bank 2 (shown in figure~\eqref{fig:capital-joint}).}
\label{fig:joint}
\vspace{-8pt}
\end{figure}
\paragraph{A.2 Joint insolvency and illiquidity}
Consider now both the cash account $V_i$ and the capital $K_i$ as they evolve over time.  As seen in the 2008 financial crisis, repurchase agreement markets can cease functioning during distress scenarios therefore the prior assumption that banks can borrow against their long-term assets may not be realistic in practice.  As such, defaults can occur if a bank becomes illiquid ($V_i(t) \leq 0$) or insolvent ($K_i(t) \leq 0$).
\[\tau_i = \inf\{t \in [0,T] \; | \; \min\{V_i(t),K_i(t)\} \leq 0\}\]
As introduced above in considering illiquidity -- defaults due to illiquidity can\emph{not} trigger additional illiquidity defaults.  However, and crucially, illiquidity \emph{can} trigger cascading failures due to insolvency.
This can be seen explicitly in Figure~\ref{fig:joint} in which the illiquidity of bank $3$ triggers the insolvency of bank $2$.  The effects of these events is shown for the solvent and liquid bank $1$; its capital drops precipitously at the joint default time, but its cash account grows due to the aforementioned effects.
\end{document}